\documentclass[twoside, 12pt]{iiser-thesis}

\usepackage{fullpage}
\usepackage{amssymb,latexsym,amsmath}     
\usepackage{graphicx}
\usepackage{authblk}
\usepackage{dsfont}
\usepackage{amsthm}
\usepackage{color}
\usepackage{float}

\newtheorem{theo}{{\bf{Theorem}}}[section]
\newtheorem{prop}[theo]{{\bf Proposition}}
\newtheorem{lem}[theo]{{\bf Lemma}}

\newtheorem{rem}{{\bf Remark}}[section]

\newtheorem{remark}[rem]{Remark}
\newtheorem{defi}{{\bf Definition}}[section]
\newtheorem{defn1}[defi]{Definition}

\renewcommand{\proof}{\noindent{\bf Proof.\ }}
\newcommand{\no}{\nonumber}
\newcommand{\noi}{\noindent}
\newcommand{\txt}{\textrm}
\newcommand{\pa}{\partial}

\newcommand{\si}{\sigma}

\newcommand{\la}{\lambda}

\newcommand{\calF}{\mathcal{F}}

\newcommand{\sig}{\sigma}
\newcommand{\tab}{\hspace*{0.3in}}

\newcommand{\vf}{\varphi}

\catcode`\@=11
\catcode`\@=12

\title{Asset Pricing in a Semi-Markov Modulated Market with Time-dependent Volatility }

\author{Tanmay S. Patankar }
\coordinator{Coordinator} \supervisor{Dr. Anindya Goswami} \sdesignation{Assistant Professor}
\department{Department of Mathematics} \reader{Prof. M. K. Ghosh}
\dedication{This thesis is dedicated to my parents.}
\graduationyear{2016}
\graduationmonth{April}

\acknowledgments{I would like to extend my sincere and heartfelt gratitude towards Dr Anindya Goswami, for his constant encouragement and able guidance at every step of my project. He has invested a great deal of time and effort in my project, and this thesis would not have been possible without his expert mentorship. I would like to thank my thesis committee member, Prof. M. K. Ghosh, and also Dr Manjunath Krishnapur of IISc Bangalore, for their valuable inputs throughout my project.
	
	\noindent I would also like to thank IISER-Pune, in particular the Department of Mathematics, for providing me all the facilities necessary for the completion of my work.

	\noindent I am extremely thankful to my family and friends for all the emotional support they have always been lending me.}

\thesisabstract{This project attempts to address the problem of asset pricing in a financial market, where the interest rates and volatilities exhibit regime switching. This is an extension of the Black-Scholes model. Studies of Markov-modulated regime switching models have been well-documented. This project extends that notion to a class of semi-Markov processes known as age-dependent processes. We also allow for time-dependence in volatility within regimes. We show that the problem of option pricing in such a market is equivalent to solving a certain integral equation.}

\begin{document}
	\thesisfront
	
%

	
	\chapter*{Introduction}

	\noi In 1973, Black, Scholes and Merton developed a mathematical model for the problem of option pricing, for which they were awarded the Nobel prize in Economics. Since then, numerous different improvements of their theoretical model are being studied. Regime switching models are one such extension of the Black-Scholes model. The goal of this project is to establish the pricing theory of defaultable	bonds for a very general kind of regime switching market. Extensive research has been done to study markets with Markov-modulated regime switching. However, it seems that the above problem with semi-Markov regimes has not yet been studied in the literature. A semi-Markov switching has past memory unlike	the well studied homogeneous Markov switching which is memoryless. Hence the former has	much greater appeal in terms of applicability than the latter. The semi-Markov switching is mathematically more interesting, too, mainly because of non-locality and unboundedness of the infnitesimal generator of the related augmented process. To address this problem, a satisfactory knowledge of continuous time stochastic processes, in particular diffusion processes and Poisson point processes, is necessary. A reasonable	understanding of pricing theory in continuous time market model is also essential.
	
	\noi We have successfully represented a large class of semi-Markov processes as solutions of a class of stochastic integral equations. This finding is original in nature and crucial to achieve the main aim of the project.
	
	\noi In the geometric Brownian motion model of asset prices, the drift and the volatility coefficients of the prices are constants. On the other hand, the regime switching model, allows those coefficients to be Markov pure jump processes. We consider a financial market where the asset price dynamics follow a regime switching model where the coefficients depend on a more general, possibly non-Markov pure jump stochastic processes. We further allow the volatility coefficient to depend on time explicitly, to capture periodic fluctuations like Monday effects etc. Under this market assumption we study locally risk minimizing pricing of vanilla options. It is shown that the price function can be obtained by solving a non-local degenerate parabolic PDE. We establish existence and uniqueness of a classical solution of at most linear growth of the PDE. We further show that the PDE is equivalent to a Volterra integral equation of second kind. Thus one can find the price function by solving the integral equation which is computationally more efficient. We finally show that the corresponding optimal hedging can be computed by performing a numerical integration.
	
	\chapter{Preliminaries}
	\begin{defi}\label{poisson def}
		Let $ (E,\mathcal{E}) $ be an Euclidean measurable space. Let $ M_P(E) $ be the set of all integer-valued measures on $(E,\mathcal{E})$. We associate $ M_P(E) $ with a $\sigma$-algebra $\mathcal{M}_P(E)$, which is the smallest $\sigma$-algebra on $ M_P(E) $ that makes the maps $ A:M_P(E)\rightarrow \mathbb{N}\cup\{0\} $, $m\mapsto m(A)$ measurable for all Borel sets $A$. Let $\mu$ be a Radon measure on $E$. A \textit{Poisson random measure} with mean measure $\mu$ is a measurable function $\wp: (\Omega,\mathcal{F},P)\rightarrow(M_P(E),\mathcal{M}_P(E))$ satisfying the following properties:
		\begin{enumerate}
			\item For $A\in\mathcal{E}$ and $k\in \mathbb{N}$, 
			\begin{equation}
			P[\omega:\wp(\omega)(A)=k]=\begin{cases}
			e^{-\mu(A)}\frac{(\mu(A))^k}{k!},\ &\mu(A)<\infty\\
			0,\ &\mu(A)=\infty.
			\end{cases}
			\end{equation}
			\item For any $m\in \mathbb{N}$, if $ A_1,A_2,\dots,A_m $ are mutually disjoint sets in $\mathcal{E}$, then $ \wp(A_1),\wp(A_2),$\\ $\dots ,\wp(A_m) $ are independent random variables.
		\end{enumerate}
	\end{defi}
	

	\begin{defi}
		A discrete-time Markov chain is a sequence of random variables $\{X_n\}_{n\geq0}$ satisfying
		\begin{equation*}
			P[X_{n+1}=x\mid X_0=x_0,X_1=x_1,\dots X_n=x_n]=P[X_{n+1}=x\mid X_n=x_n],
		\end{equation*}
		provided both conditional probabilities are well-defined, i.e $P[X_0=x_0,X_1=x_1,\dots X_n=x_n]>0$.
	\end{defi}
	
	\begin{defi}
		A continuous-time time-homogeneous Markov chain with rate matrix $ \Lambda $ is a stochastic process $\{X_t\}_{t\geq 0}$ satisfying the following conditions
		\begin{enumerate}
			\item $ X_t $ is a piecewise constant right-continuous process with left-limits, with discontinuities at a discrete set $\{T_n\}_{n\geq1}$. (This means that $X_t$ is a right-continuous process whose left-hand limit exists at all points with probability 1.)
			\item The sequence $\{X_{T_n}\}_{n=0,1,\dots}$ is a Markov chain with transition matrix $\mathbb{P}=(p_{ij})$, where $p_{ij}=\frac{\lambda_{ij}}{|\lambda_{ii}|}$.
			\item $P\left[ X_{T_{n+1}}=j, T_{n+1}-T_n\leq y |(X_0,T_0),(X_1,T_1),\dots,(X_{T_n}=i,T_n) \right]= p_{ij}(1-e^{\lambda_{ii} y})$.
		\end{enumerate}
	\end{defi}
	
	\begin{defi}
		A general continuous-time Markov process is a process $\{X_t\}_{t\geq 0}$ on a probability space $(\Omega,\mathcal{F},P)$ and taking values in a measurable space $(S,\mathcal{S})$, satisfying 
		\begin{equation}
			P[X_t\in A\mid \mathcal{F}_s]=P[X_t\in A\mid X_s]
		\end{equation}
		for all $A\in\mathcal{S}$  and for each $s<t$.
	\end{defi}
	
	\begin{defi}
		A \textit{semi-Markov} process is a process $\{X_t\}_{t\geq 0}$ that satisfies the following properties:
		\begin{enumerate}
			\item $ X_t $ is a piecewise constant rcll process with discontinuities at a discrete set $\{T_n\}_{n\geq1}$. 
			\item The transition probabilities satisfy
			\begin{align}
			& \no P\left[ X_{T_{n+1}}=j, T_{n+1}-T_n\leq y |(X_0,T_0),(X_1,T_1),\dots,(X_{T_n},T_n) \right]\\
			=& P\left[ X_{T_{n+1}}=j, T_{n+1}-T_n\leq y |X_{T_n} \right].
			\end{align}
		\end{enumerate}
	\end{defi}

	\begin{defi}
		A $C_0$-semigroup of operators $\{S(t)\}_{t\geq0}$ on a Banach space $ V $ is a map $S:\mathbb{R}_+\rightarrow BL(V)$, such that 
		\begin{enumerate}
			\item $ S_0f=f \quad  \forall f\in V $,
			\item $ S_{t+s}=S_t\circ S_s \quad \forall t,s\geq 0$, and
			\item $ \| S_t f - f \|\rightarrow0 $ as $t\downarrow0$, for all $ f\in V $.
		\end{enumerate}
	\end{defi}
	
	\begin{defi}
		Let $\{S(t)\}_{t\geq0}$ be a $C_0$-semigroup of operators. The domain of the infinitesimal generator of the semigroup is defined as
		\begin{equation*}
			\mathcal{D}(\mathcal{A}):=\left\lbrace f\in V \mid \lim_{t\rightarrow 0}\frac{S_tf-f}{t} \textrm{ exists} \right\rbrace 
		\end{equation*}
		and the infinitesimal generator of $f$ is  the operator $\mathcal{A}$, defined such that
		\begin{equation*}
			\mathcal{A}f:=\lim_{t\rightarrow 0}\frac{S_tf-f}{t}
		\end{equation*}
		for all $f\in\mathcal{D}$.
	\end{defi}

	\chapter{Age-dependent processes}
	\section{Time-homogeneous Age-dependent processes}\label{sec 2.1}
	\noi We consider a class of stochastic processes which is constructed as a strong solution of a certain set of stochastic integral equations. Let $(\Omega,\mathcal{F},\{\mathcal{F}_t\}_{t\geq0},P)$ be a filtered probability space, and $\chi=\left\lbrace 1,2,\dots,k\right\rbrace\subset\mathbb{R} $ be the state space. For $i,j\in \chi$ and $i\neq j$,  define 
	\begin{equation}\label{lambda from}
		\lambda:\chi\times\chi\times\ (0,\infty)\rightarrow[0,\infty)
	\end{equation}
	to be a measurable function with 
	\begin{equation}\label{lambda condition 1}
	\sup_{y\in(0,\infty)}\sum_{j\neq i}\lambda_{ij}(y)<\infty.
	\end{equation} and
	\begin{equation}\label{lambda condition 2}
		\lim\limits_{y\rightarrow\infty} \Lambda_{i}(y)=\infty, \text{\ where\ } \Lambda_i(y)=\int_{0}^{y}\sum_{j\neq i} \lambda_{ij}(v)\,dv.
	\end{equation}
	
	\noi The diagonal elements  are defined as $\lambda_{ii}(y):=-\sum_{j\neq i} \lambda_{ij}(y)$.

	\noi For $i\neq j, y>0$,  let $\Lambda_{ij}(y)$ be consecutive (w.r.t the lexicographical ordering) right-open, left-closed intervals of length $\lambda_{ij}(y)$. Define $h:\ \chi\times\mathbb{R}_+\times\mathbb{R}\rightarrow\mathbb{R}$ as
	\begin{equation}
		h(i,y,z)=\sum_{j\neq i\in \chi}(j-i)\mathds{1}_{\Lambda_{ij}(y)}(z) \label{hdef}
	\end{equation}
	and a function $g:\ \chi\times\mathbb{R}_+\times\mathbb{R}\rightarrow\mathbb{R}$ as
	\begin{equation}
		g(i,y,z)=y\sum_{j\neq i\in \chi}\mathds{1}_{\Lambda_{ij}(y)}(z). \label{gdef}
	\end{equation}.

	\noi We consider the following system of coupled stochastic integral equations in $X_t$ and $Y_t$:
	\begin{eqnarray}
		X_t=X_0+\int_{(0,t]}\int_{\mathbb{R}}h(X_{u-},Y_{u-},z)\,\wp(du,dz)\label{xdef}\\
		Y_t=Y_0+t-\int_{(0,t]}\int_{\mathbb{R}}g(X_{u-},Y_{u-},z)\,\wp(du,dz),\label{ydef}
	\end{eqnarray}
	where $h$ and $g$ are defined by equations \eqref{hdef} and \eqref{gdef} respectively, $\wp(du,dz)$ is a Poisson random measure on $ \mathbb{R}_+\times\mathbb{R} $ with intensity $du\times dz$, and $ \{\wp((0.t]\times dz)\}_{t\geq0} $ is adapted to the filtration $ \{\mathcal{F}_t\}_{t\geq0} $.
	
	\begin{theo}\label{exun1}
		There exists a unique strong solution to equations \eqref{xdef} and \eqref{ydef}.
	\end{theo}
	\begin{proof}First, we note that \eqref{lambda condition 1} can be rewritten as:
		\begin{equation}
			\sum_{j\neq i}\lambda_{ij}(y)<c\ \text{for all } y\in[0,\infty),
		\end{equation}
		for some $c>0$. Thus, it follows that
		\begin{equation*}
			\bigcup_{y\in(0,T]}\left[ \{y\}\times [0,\sum_{j\neq i}\lambda_{ij}(y)]\right] \subset [0,T]\times[0,c].
		\end{equation*}
		
		\noi The interval $[0,c]$ has finite Lebesgue measure $c$. Define $D$ to be the set of all point masses of the measure $\wp(\omega)$:
		\begin{equation*}
			D:=\left\lbrace s\in(0,\infty)|\wp(\omega)(\{s\}\times[0,c])=1 \right\rbrace \text{ for any fixed $ \omega\in\Omega $}.
		\end{equation*}
		
		\noi An illustration of a sample of the points of a Poisson random measure with $ c=10 $ and $T=1$ is shown in Figure \ref{fig:poisson_points}.
		
		\begin{figure}
			\centering
			\includegraphics[width=0.7\linewidth]{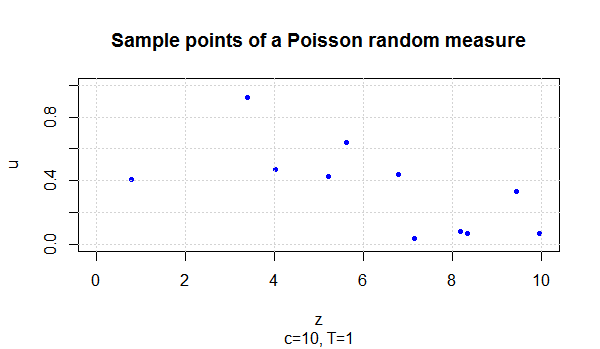}
			\caption{A sample of point masses of a Poisson random measure}
			\label{fig:poisson_points}
		\end{figure}

		\noi The set of all transition times of $X(\omega)$ is a subset of $D$.
		Since the measure of the set $ [0,c] $ is finite, $D$ is a discrete set (i.e, $D$ has no limit point) with probability 1. We can thus enumerate the set $D$ as
		\begin{equation*}
			D=\{0,\sigma_1,\sigma_2,\dots, \},
		\end{equation*}
		and it is easy to see that $ \sigma_1,\sigma_2,\dots $ are stopping times under the filtration $\mathcal{F}_t$ of the underlying probability space. Since $D$ is a discrete set, 
		\begin{equation}\label{sigma infty}
			\lim_{n\rightarrow\infty}\sigma_n = \infty \ \text{a.s.}
		\end{equation}
		We use an iterative argument for proving the existence and uniqueness of a strong solution to \eqref{xdef} and \eqref{ydef}. For a fixed $\omega$, we construct a  solution to this pair of equations on the time interval $[0,\sigma_1]$. Then we extend this solution to the time interval $(\sigma_1,\sigma_2]$, and so on.
		
		\noi Since 
		\begin{equation*}
			\wp(\omega)([0,\sigma_1)\times [0,c])=0,
		\end{equation*}
		for $ t\in[0,\sigma_1) $,
		\begin{align*}
			X_t(\omega)&=X_0+\int_{0}^{t}\int_{[0,c]}h(X_{u-},Y_{u-},z)\,\wp(\omega)(du,dz)\\
			&=X_0
		\end{align*}
		and
		\begin{align*}
			Y_t(\omega)&=Y_0+t-\int_{0}^{t}\int_{[0,c]}g(X_{u-},Y_{u-},z)\,\wp(\omega)(du,dz)\\
			&=Y_0+t.
		\end{align*}
		At $t=\sigma_1$, 
		\begin{align*}
			X_{\sigma_1}(\omega)&=X_0+\int_{[0,c]}h(X_0,Y_0+\sigma_1,z)\,\wp(\omega)(\{\sigma_1\}\times dz),\\
			Y_{\sigma_1}(\omega)&=Y_0+\sigma_1-\int_{[0,c]}g(X_0,Y_0+\sigma_1,z)\,\wp(\omega)(\{\sigma_1\}\times dz).
		\end{align*}
		Since we have been able to write down the solution in the time interval $[0,\sigma_1]$ explicitly, it is obviously unique.
		
		\noi Now we consider the time interval $(\sigma_1,\sigma_2]$. We define the following quantities:
		\begin{align*}
			\tilde{X}(0)&=X(\sigma_1),\\
			\tilde{Y}(0)&=Y(\sigma_1),\\
			\tilde{\wp}(ds,dz)&=\wp(\sigma_1+ds,dz),\\
			\tilde{\sigma}_n&=\sigma_{n+1}-\sigma_1.
		\end{align*} 
		Then, $ \tilde{D}=\left\lbrace s>0|s+\sigma_1\in D \right\rbrace = \{\tilde{\sigma}_n\}_{n\geq1} . $
		Now we consider the equations \eqref{xdef} and \eqref{ydef} on $[0,\tilde{\sigma_1}]$, where $X_0, Y_0$ and $\wp$ are replaced by $\tilde{X}(0), \tilde{Y}(0)$ and $\tilde{\wp}$ respectively. If $t\in[0,\tilde{\sigma_1})$, then the solution $\tilde{X}_t$ is given by
		\begin{align*}
			\tilde{X}_t&=\tilde{X}_0,\\
			\tilde{Y}_t&=\tilde{Y}_0+t.
		\end{align*}
		and for $t=\tilde{\sigma}_1$, we have:
		\begin{align*}
			\tilde{X}_{\tilde{\sigma}_1}&=\tilde{X}_0+\int_{[0,c]}h(\tilde{X}_0,\tilde{Y}_0+\tilde{\sigma}_1,z)\,\tilde{\wp}(\{\tilde{\sigma}_1\}\times dz),\\
			\tilde{Y}_{\tilde{\sigma}_1}&=\tilde{Y}_0+\tilde{\sigma}_1-\int_{[0,c]}g(\tilde{X}_0,\tilde{Y}_0+\tilde{\sigma}_1,z)\,\tilde{\wp}(\{\tilde{\sigma}_1\}\times dz).
		\end{align*}
		Therefore, the solution $(X,Y)$ of the original equations can be reconstructed from $(\tilde{X},\tilde{Y})$ by the following relation
		\begin{align*}
			X_t&=\begin{cases}
				X(t),\ t\in[0,\sigma_1]\\
				\tilde{X}(t-\sigma_1),\ t\in(\sigma_1,\sigma_2]
			\end{cases}\\
			Y_t&=\begin{cases}
				Y(t),\ t\in[0,\sigma_1]\\
				\tilde{Y}(t-\sigma_1),\ t\in(\sigma_1,\sigma_2].
			\end{cases}
		\end{align*}
		This establishes the existence and uniqueness of the strong solution in the time interval $[0,\sigma_2]$.
		
		\noi Continuing in this fashion, we can uniquely construct the solution in successive time intervals. By \eqref{sigma infty}, this sequence of intervals covers the entire positive real time-axis. Hence, the solution is globally determined.	\qed
	\end{proof}\bigskip
	\noi In the above proof, it is evident that the process $X_t$ has almost surely piecewise constant r.c.l.l paths. The points of discontinuity of $X_t$ are called transition times.
	
	\begin{defi}
		Transition times are elements of an increasing sequence $ \{T_n\}_{n\geq 1} $ such that $ \{T_n:n\geq 1\}=\{t> 0: X_t\neq X_{t-}\} $. We set $T_0:=-Y_0$.
		We define the holding times $\tau_n:=T_n-T_{n-1}$ for all $n\geq1$.		
	\end{defi}\bigskip
	
	\noi From the above definitions, it is clear that
	\begin{equation*}
	X_u-X_{u-}=\int_{\mathbb{R}}h(X_{u-},Y_{u-},z)\,\wp(\{u\}\times dz)
	\end{equation*}
	is non-zero if and only if $u=T_n$ for some positive integer $n$. This also implies that
	\begin{equation*}
	\int_{\mathbb{R}}g(X_{u-},Y_{u-},z)\,\wp(\{u\}\times dz)=\begin{cases}
	Y_{u-},\text{ if $u=T_n$ for some $n$}\\
	0,\text{ otherwise}.
	\end{cases}
	\end{equation*}
	Hence, by induction, we obtain, for any integer $n\geq 0$,
	\begin{equation*}
	Y_t=Y_0+t-\sum_{r=1}^n Y_{T_r-} \text{ for } t\in[T_n,T_{n+1}).
	\end{equation*}
	Thus, $Y_t=0$ iff $t=T_n$ for some $n\in\mathbb{N}$, and $Y_{T_n-}=T_n-T_{n-1}$ for all $n\in\mathbb{N}$. This observation motivates us to define the following:
	
	\begin{defi}
		Let $(X_t,Y_t)$ be the unique strong solution to equations \eqref{xdef} and \eqref{ydef}. The process $X_t$ is then called an ``age-dependent process'' and  $Y_t$ is called the ``holding time process'' corresponding to $X_t$.
	\end{defi}

	\begin{theo}
		Let $(X_t,Y_t)$ be the unique strong solution to equations \eqref{xdef} and \eqref{ydef}. The process $Z_t:=(X_t,Y_t)$ is a Markov process.
	\end{theo}
	\begin{proof} From equations \eqref{xdef} and \eqref{ydef}, we get, for $t<T$,
		\begin{align*}
			{X}_T&={X}_0+\int_{0}^{T}\int_{\mathbb{R}}h({X}_{u-},{Y}_{u-},z)\,{\wp}(du,dz)\\
			&={X}_t+\int_{t}^{T}\int_{\mathbb{R}}h({X}_{u-},{Y}_{u-},z)\,{\wp}(du,dz)
		\end{align*}
		and
		\begin{align*}
			{Y}_T&=Y_0+T-\int_{0}^{T}\int_{\mathbb{R}}g({X}_{u-},{Y}_{u-},z)\,{\wp}(du,dz)\\
			&={Y}_t+(T-t)-\int_{t}^{T}\int_{\mathbb{R}}g({X}_{u-},{Y}_{u-},z)\,{\wp}(du,dz).
		\end{align*}
		From the second property of $ \wp $, as in Definition \ref{poisson def} and from the above expressions, it is thus clear that $Z_t$ is a Markov process.\qed
	\end{proof}
	
	\noi It is also easy to see that $Z_t$ is strongly Markov.


	\begin{theo}\label{agedep}
		Let $X_t$ be an age-dependent process. Then, $X_t$ is a semi-Markov process.
	\end{theo}
	\begin{proof}We have already seen in the proof of Theorem \ref{exun1}, that $X_t$ is a piecewise constant right-continuous process, and the left-hand limits exist. In other words, $ X_t $ is a c\`{a}dl\`{a}g process.
		

\noi Next, we show that 
\begin{align}
& \no P[X_{T_{n+1}}=j,T_{n+1}-T_n\leq y\mid (X_{T_0},{T_0}),(X_{T_1},{T_1}),\dots,(X_{T_n},{T_n})].\\
=& P[X_{T_{n+1}}=j,T_{n+1}-T_n\leq y\mid X_{T_n}].\label{note 1}
\end{align}

\noi We note that the LHS of \eqref{note 1} can be written as
\begin{align}
	& \no P(X_{T_{n+1}}=j\mid (X_{T_0},{T_0}),(X_{T_1},{T_1}),\dots,(X_{T_n},{T_n}),\{ T_{n+1}-T_n\leq y \})\times\\
	& P[T_{n+1}-T_n\leq y\mid (X_{T_0},{T_0}),(X_{T_1},{T_1}),\dots,(X_{T_n},{T_n})]\label{note 2}
\end{align}

\noi From equation \eqref{xdef},
\begin{align*}
	X_{T_{n+1}}=X_{T_n}+\int_{\mathbb{R}}h(X_{T_n},T_{n+1}-T_n,z)\,\wp(\{T_n+(T_{n+1}-T_n)\}\times dz),
\end{align*}
since $ Y_{T_{n+1}\!-}=T_{n+1}-T_n $ and
\begin{equation*}
	\int_{(T_n,T_{n+1})}\int_{\mathbb{R}}h(X_{u-},Y_{u-},z)\,\wp(du,dz)=0.
\end{equation*}

\noi Again, since $ \wp $ is a Poisson random measure, for any Borel set $ B\subset (0,\infty)\times\mathbb{R} $, $ \wp((T_n,0)+B) $ is independent of $ \mathcal{F}_{T_n} $. Therefore,
\begin{align}
	\no & P(X_{T_{n+1}}=j\mid (X_{T_0},{T_0}),(X_{T_1},{T_1}),\dots,(X_{T_n},{T_n}),\{ T_{n+1}-T_n\leq y \})\\
	=& \no P\bigg(\int_{\mathbb{R}}h(X_{T_n},T_{n+1}-T_n,z)\,\wp(\{T_n+(T_{n+1}-T_n)\}\times dz)=j-X_{T_n}\Big|\bigg.\\
	&\bigg.\no (X_{T_0},{T_0}),(X_{T_1},{T_1}),\dots,(X_{T_n},{T_n}),\{ T_{n+1}-T_n\leq y \}    \bigg)\\
	=& \no P\left(\int_{\mathbb{R}}h(X_{T_n},T_{n+1}-T_n,z)\,\wp(\{T_n+(T_{n+1}-T_n)\}\times dz)=j-X_{T_n}\Big|X_{T_n},T_n, \{ T_{n+1}-T_n\leq y \}  \right) \\
	=& P\left(X_{T_{n+1}}=j\mid X_{T_n},\{ T_{n+1}-T_n\leq y \} \right), \label{note 3}
\end{align}
since the distribution of $ \wp(B) $ depends only on the Lebesgue measure of $ B $ and thus is invariant under the translation of $ B $. 

\noi For every $ \omega\in\Omega $, equation \eqref{ydef} implies that
\begin{equation*}
	\int_{(T_n,T_n+t]}\int_{\mathbb{R}}g(X_{T_n},u-T_n,z)\,\wp(du,dz)=\begin{cases}
	0, \text{ for } t<T_{n+1}-T_n\\
	T_{n+1}-T_n,\text{ for } t= T_{n+1}-T_n.
	\end{cases}
\end{equation*}

\noi Hence, $ T_{n+1}-T_n $ is the first non-zero value of the following map
\begin{equation*}
	t\mapsto\int_{(0,t]}\int_{\mathbb{R}}g(X_{T_n},u,z)\,\wp(T_n+du,dz).
\end{equation*}
Again, since $ \wp(T_n+du,dz) $ is independent of $ \mathcal{F}_{T_n} $ and $ T_n $, we obtain, from the above,
\begin{align}
	&\no P[T_{n+1}-T_n\leq y\mid (X_{T_0},{T_0}),(X_{T_1},{T_1}),\dots,(X_{T_n},{T_n})]\\
	=& P[T_{n+1}-T_n\leq y\mid X_{T_n}]. \label{note 4}
\end{align}
Thus, using \eqref{note 2}, \eqref{note 3} and \eqref{note 4}, the LHS of \eqref{note 1} is equal to
\begin{align*}
	& P(X_{T_{n+1}}=j\mid X_{T_n},\{ T_{n+1}-T_n\leq y \})\times P[T_{n+1}-T_n\leq y\mid X_{T_n}]\\
	=& P(X_{T_{n+1}},T_{n+1}-T_n\leq y\mid X_{T_n})\\
	=& \text{RHS of equation \eqref{note 1}}.
\end{align*}

\noi Hence, $X_t$ is a semi-Markov process.\qed

	\end{proof}
	
	\noi We define a function $F:[0,\infty)\rightarrow [0,1]$ as $F(y|i):=1-e^{-\Lambda_i(y)}$. From \eqref{lambda condition 2}, $ \Lambda_i(y) $ is an absolutely continuous function of $y$. Thus, $F(y|i)$ is differentiable almost everywhere. Let $ f(y|i):=\frac{d}{dy}F(y|i) $. We also  define $p_{ij}(y)$, such that
	\begin{equation}\label{p_ij}
		p_{ij}(y):=\begin{cases}\dfrac{\lambda_{ij}(y)}{-\lambda_{ii}(y)}\mathds{1}_{(0,\infty)}(-\lambda_{ii}(y)),\ j\neq i\\
		\mathds{1}_{\{0\}}(-\lambda_{ii}(y)),\ j=i.
				\end{cases}
	\end{equation}
	This ensures that $ [p_{ij}(y)] $ is a probability matrix for all $y$.
	
	
	\begin{prop}
		\begin{enumerate}
			\item The function $F$ is the conditional c.d.f of the holding time of the age-dependent process $X_t$.
			\item $ p_{ij}(y)=P[X_{T_{n+1}}=j|X_{T_n}=i, Y_{T_{n+1}-}=y] $.
		\end{enumerate}
	\end{prop}
	\begin{proof}
		The conditional c.d.f of the holding time after the $n$\textsuperscript{th} transition, given the $n$\textsuperscript{th} state, is
		\begin{align*}
		& P[\tau_{n+1}\leq y|X_{T_n}=i]\\
		&=1-P[\{\text{No transition in }(T_n, T_n+y]\}|X_{T_n}=i]\\
		&=1-P[\wp\{(u,z)\in \mathbb{R}_+ \times \mathbb{R}_+|z\in \bigcup_{j\neq i}\Lambda_{ij}(u)\}=0\mid X_{T_n}=i],\text{ where $ u=s+T_n $ and $ s\in (0,y) $}\\
		&=1-e^{-\Lambda_i(y)}\\
		&=F(y|i).
		\end{align*}
		
		\noi Also, we note that, for $j\neq i$, $ P[X_{T_{n+1}}=j|X_{T_n}=i, Y_{T_{n+1}-}=y] $ is the probability of the event that  a Poisson point mass lies somewhere in  $\{\tau_n+y\}\times\Lambda_{ij}(y)$, given no transition of $X_t$ occurs within time $y$. This probability is
		\begin{align*}
		&\dfrac{|\Lambda_{ij}(y)|}{|\bigcup_{j\neq i}\Lambda_{ij}(y)|}\mathds{1}_{(0,\infty)}(|\bigcup_{j\neq i}\Lambda_{ij}(y)|)\\
		=&\dfrac{\lambda_{ij}(y)}{-\lambda_{ii}(y)}\mathds{1}_{(0,\infty)}(-\lambda_{ii}(y))\\
		=& p_{ij}(y).
		\end{align*}\qed
	\end{proof}
	\noi We note that under the assumptions \eqref{lambda condition 1} and \eqref{lambda condition 2}, $ F(y|i)<1 $ for all $ y>0 $ and $ \lim_{y\rightarrow\infty}F(y|i)=1 $. Thus, the holding times are unbounded but finite almost surely.

	\begin{prop}
		We have, for $ y>0 $,
		\begin{equation*}
		p_{ij}(y)\dfrac{f(y|i)}{1-F(y|i)}=\begin{cases}
		\lambda_{ij}(y),\textrm{ for } i\neq j,\\
		0,\textrm{ for } i=j.
		\end{cases}
		\end{equation*}
	\end{prop}
	\begin{proof}
		
		
		
		\begin{equation*}
		F(y|i):=1-e^{-\Lambda_i(y)}.
		\end{equation*}
		Hence, differentiating w.r.t $y$, we have
		\begin{align}
			\no f(y|i)&=-\lambda_{ii}(y)e^{-\Lambda_i(y)}\\
			\frac{f(y|i)}{1-F(y|i)}&=-\lambda_{ii}(y)\label{lamb ii}.
		\end{align}
		Hence, for $ i\neq j $,
		\begin{align}
		\no p_{ij}(y)\dfrac{f(y|i)}{1-F(y|i)} =& -\lambda_{ii}(y)\times\dfrac{\lambda_{ij}(y)}{-\lambda_{ii}(y)}\mathds{1}_{(0,\infty)}(-\lambda_{ii}(y))\\
		\no =& \lambda_{ij}(y),
		\end{align}
		since if $ \lambda_{ii}(y)=0 $, then for each $ j(\neq i),\ \lambda_{ij}(y)=0 $. Again, if $ \lambda_{ii}(y)=0 $, then $ p_{ii}(y)=0 $ and if if $ \lambda_{ii}(y)=0 $, then $ \frac{f(y|i)}{1-F(y|i)}=0 $ from \eqref{lamb ii}. Thus
		\begin{align*}
		\no p_{ii}(y)\dfrac{f(y|i)}{1-F(y|i)} =& 0\\
		\end{align*}
		for all $ y>0 $.
		\qed
	\end{proof}
	
	\noi We can also easily verify, from \eqref{p_ij}, that $\sum_{j\in\chi} p_{ij}(y)=1$. 
\begin{theo}
		Let $X_t$ be an age-dependent process as described in equations \eqref{xdef} and \eqref{ydef}. Then, its kernel is given by (for $ y>0,\ i\neq j $) 
		\begin{equation*}
		Q_{ij}(y):=P[X_{T_{n+1}}=j, Y_{T_{n+1}-}\leq y | X_{T_n}=i]
		=\int_{0}^{y}e^{-\Lambda_i(s)}\lambda_{ij}(s)\,ds.
		\end{equation*}
\end{theo}
	\begin{proof}
		We note that
		\begin{align*}
			Q_{ij}(y):&= P[X_{T_{n+1}}=j, Y_{T_{n+1}-}\leq y | X_{T_n}=i]\\
			&=E[P(X_{T_{n+1}}=j, Y_{T_{n+1}-}\leq y|X_{T_n}=i,Y_{T_{n+1}-})|X_{T_n}=i]\\
			&=\int_{0}^{\infty}\mathds{1}_{[0,y]}(s)P[X_{T_{n+1}}=j|X_{T_n}=i,Y_{T_{n+1}-}=s]f(s|i)\,ds\\
			&=\int_0^y p_{ij}(s)f(s|i)\,ds\\
			&=\int_0^y (1-F(s|i))\lambda_{ij}(s)\,ds\\
			&=\int_{0}^{y}e^{-\Lambda_i(s)}\lambda_{ij}(s)\,ds.
		\end{align*}
		\qed
	\end{proof}
	

	\noi It seems that in the literature, for the first time, this class of processes appears as ``Age-dependent processes'' in \cite{agedep}. In \cite{naval}, the class of semi-Markov processes is studied after dividing it into two categories, namely type I and type II. We recognise that the age-dependent process being discussed in this chapter belongs to type II. Here, we present the hierarchy of some important classes of pure jump processes in continuous time.
	\begin{gather*}
	\text{Pure jump processes}\\
	\cup\\
	\text{Time-homogeneous case}\\
	\cup\\
	\text{Semi-Markov processes}\\
	\cup\\
	\text{Age-dependent case}\\
	\cup\\
	\text{Age-independent case}\\
	\cup\\
	\text{Markov processes}
	\end{gather*}

	\section{Time-inhomogeneous Age-dependent processes}
	It is interesting to note that the construction of age-dependent processes in Section \ref{sec 2.1} can easily be generalised to construct a time-inhomogeneous non-Markov pure jump process. To this end we consider a Poisson random measure $N$ has the form
	
	\begin{equation*}
		N(dt,dz):= \wp(d\eta(t),dz),
	\end{equation*}
	where $ \eta $ is an increasing differentiable function with $ \eta(0)=0 $. 
	
	\noi This random measure has intensity $\eta'(t)\,dt\,dz$, where $\eta'$, under the assumption, is a continuous function from $[0,\infty)$ to $(0,\infty)$. Thus, 
	\begin{equation*}
		E[N(A)]=\int_A \eta(t)dt\,dz
	\end{equation*}
	for any set $A\in \mathcal{F}$. We consider a new pair of coupled stochastic integral equations in $ (\tilde{X}_t,\tilde{Y}_t) $:
	\begin{eqnarray}
		\tilde{X}_t=\tilde{X}_0+\int_{0}^{t}\int_{\mathbb{R}}h(\tilde{X}_{u^-},\tilde{Y}_{u^-},z)\,N(du,dz)   \label{xtildef}\\
		\tilde{Y}_t=\tilde{Y}_0+t-\int_{0}^{t}\int_{\mathbb{R}}g(\tilde{X}_{u^-},\tilde{Y}_{u^-},z)\,N(du,dz),\label{ytildef}
	\end{eqnarray}
	where $g$ and $h$ are defined by equations \eqref{gdef} and \eqref{hdef}, respectively.
	
	
	\begin{theo}
		There exists a unique solution $ (\tilde{X}_t,\tilde{Y}_t) $ to equations \eqref{xtildef} and \eqref{ytildef}.
	\end{theo}
	\begin{proof}
		The proof can be constructed in a similar way as that of Theorem \ref{exun1}.
	\end{proof}
	
	\begin{theo}
		The process $\tilde{Z}_t:=(\tilde{X}_t,\tilde{Y}_t)$ is a Markov process.
	\end{theo}
	\begin{proof}
		
		\begin{align*}
			\tilde{X}_T&=\tilde{X}_0+\int_{0}^{T}\int_{\mathbb{R}}h(\tilde{X}_{u^-},\tilde{Y}_{u^-},z)\,N(du,dz)\\
			&=\tilde{X}_t+\int_{t}^{T}\int_{\mathbb{R}}h(\tilde{X}_{u^-},\tilde{Y}_{u^-},z)\,N(du,dz)
		\end{align*}
		and
		\begin{align*}
			\tilde{Y}_T&=\tilde{Y}_0+T-\int_{0}^{T}\int_{\mathbb{R}}g(\tilde{X}_{u^-},\tilde{Y}_{u^-},z)\,N(du,dz)\\
			&=\tilde{Y}_t+(T-t)-\int_{t}^{T}\int_{\mathbb{R}}g(\tilde{X}_{u^-},\tilde{Y}_{u^-},z)\,N(du,dz).
		\end{align*}
		
	\end{proof}
	
	\begin{theo}
		$\tilde{Z}_t:=(\tilde{X}_t,\tilde{Y}_t)$ is a c\`{a}dl\`{a}g process.
	\end{theo}
	\begin{proof}
		This follows from the fact that $\eta'$ is bounded on compact sets and $ \lambda_{ij}(y) $ being bounded from equation \eqref{lambda condition 1}. 
	\end{proof}
	
	\begin{theo}
		The sequence $\{\tilde{X}_{T_n}\}_n$ is a Markov chain.
	\end{theo}
	\begin{proof}
		\begin{align*}
			& P[\tilde{X}_{T_{n+1}}=j|\tilde{X}_{T_0},\tilde{X}_{T_1},\dots,\tilde{X}_{T_n}=i]	\\
			=& E[P(\tilde{X}_{T_{n+1}}=j|\mathcal{F}_{T_n},T_{n+1},\tilde{X}_{T_n})|\tilde{X}_{T_0},\tilde{X}_{T_1},\dots,\tilde{X}_{T_n}=i]\\
			=& E(P[N(\left\lbrace T_{n+1} \right\rbrace\times\Lambda_{ij}(\tau_{n+1}) )\neq 0|N(\left\lbrace T_{n+1} \right\rbrace\times\Lambda_{ik}(\tau_{n+1}))\neq 0 \mathrm{\ for\ some\ } k]|\tilde{X}_{T_0},\tilde{X}_{T_1},\dots,\tilde{X}_{T_n}=i)\\
			=& E\left[ \dfrac{\lambda_{ij}(T_{n+1}-T_n)}{\sum_{k\neq j}\lambda_{ik}(T_{n+1}-T_n)} \bigg|\tilde{X}_{T_0},\tilde{X}_{T_1},\dots,\tilde{X}_{T_n}=i,\tilde{Y}_{T_n}=0\right] \\
			=& E\left[ \dfrac{\lambda_{ij}(\tilde{Y}_{T_{n+1}\!-})}{\sum_{k\neq j}\lambda_{ik}(\tilde{Y}_{T_{n+1}\!-})} \bigg|\tilde{X}_{T_n}=i\right],
		\end{align*}
		since the conditional distribution of $ \tilde{Y}_{T_{n+1}\!-} $ given $ \mathcal{F}_{T_n} $ is the same as that given $ \tilde{X}_{T_n} $. Thus the conditional probability on the LHS depends entirely on $\tilde{X}_{T_n}=i$.\qed
	\end{proof}

	\noi However, the process $\tilde{X}_t$ is not a semi-Markov process. This is because the transition probability can be written as
	\begin{align*}
		& P\Bigg[ \bigcup_{0<s'\leq y}\left\lbrace N\left( \bigcup_{0<s< s'}\left\lbrace T_n+s\right\rbrace \times \bigcup_{k\neq i}\Lambda_{ik}(s) \right)=0  \right\rbrace \bigcap\left\lbrace N\left( \left\lbrace T_n+s'\right\rbrace \times \Lambda_{ij}(s') \right)=1  \right\rbrace \Bigg| \mathcal{F}_{T_n},\tilde{X}_{T_n}=i \Bigg] \\
		=& P\Bigg[ \bigcup_{0<s'\leq y}\left\lbrace N\left( \bigcup_{0<s< s'}\left\lbrace T_n+s\right\rbrace \times \bigcup_{k\neq i}\Lambda_{ik}(s) \right)=0  \right\rbrace \bigcap\left\lbrace N\left( \left\lbrace T_n+s'\right\rbrace \times \Lambda_{ij}(s') \right)=1  \right\rbrace \Bigg| \tilde{X}_{T_n}=i,T_n \Bigg].
	\end{align*}
	However, the Poisson random measure $ N $ is not translation-invariant with respect to time, unless $\eta'$ is a constant. Hence, no further simplification is possible in general.\qed
	
	\section{The Infinitesimal Generator}
	We will derive an expression for the infinitesimal generator of an augmented age-dependent process. Let $(X_t,Y_t)$ be an augmented age-dependent process. Let $\phi: \chi\times[0,\infty)$ be a differentiable function. Then, by It\={o}'s formula,
	\begin{align}
	d\phi(X_t,Y_t)&=\frac{\partial \phi}{\partial y}(X_t,Y_t)\,dY_t^c + \phi(X_t,Y_t)-\phi(X_t-,Y_t-)\nonumber\\
	&=\no \frac{\partial \phi}{\partial y}(X_t,Y_t)\,dt\\
	&+\int_{\mathbb{R}}\left[ \phi(X_{t-}+h(X_{t-},Y_{t-},z),Y_{t-}-g(X_{t-},Y_{t-},z))-\phi(X_{t-},Y_{t-}) \right](\hat{\wp}(dt,dz)+dt\,dz),
	\end{align} 
	where $ \hat{\wp}(dt,dz)={\wp}(dt,dz)-dt\,dz $ is the compensated Poisson random measure, with mean zero, independent of $X_0$. The process obtained by integrating w.r.t $\hat{\wp}$ is a martingale, $M_t$. Hence, we can write
	\begin{equation}
	d\phi(X_t,Y_t)=\frac{\partial \phi}{\partial y}(X_t,Y_t)\,dt + \sum_{j\neq X_{t-}}[\phi(j,0)-\phi(X_{t-},Y_{t-})]\lambda_{X_{t-}j}(Y_{t-})\,dt+dM_t.
	\end{equation}
	Thus, the infinitesimal generator, $\mathcal{L}$, of the augmented age-dependent process is given by the following expression:
	\begin{equation}
		\mathcal{L}\phi(i,y)=\frac{\partial \phi}{\partial y}(i,y)+\sum_{j\neq i}[\phi(j,0)-\phi(i,y)]\lambda_{ij}(y).
	\end{equation}

	\section{An example}
	Here we present some example of age-dependent processes with finitely many states. We let the (age-dependent) transition rate matrix be given by
	\begin{equation}\label{example1}
		\Lambda(y)=\Lambda^{(1)}+y\Lambda^{(2)},
	\end{equation}
	where $ \Lambda^{(1)} $ and $ \Lambda^{(2)} $ are two rate matrices of order $ k $. If, in a particular case, $ \Lambda^{(2)}=0 $, the trivial matrix, then $ \Lambda(y)=\Lambda^{(1)} $ for all $ y $ and the resulting process becomes Markov. Whereas, the resulting process becomes an age-independent semi-Markov process when 
	$ \Lambda^{(1)}=c\Lambda^{(2)}$ for some $c\in \mathbb{R}_+ $. But of course, in general, $ \Lambda(y) $ prescribes an age-dependent process.
	
	\noi The transition probabilities for this process are given by $ (i\neq j) $
	\begin{equation}
		p_{ij}(y)=\frac{\lambda^{(1)}_{ij}+y\lambda^{(2)}_{ij}}{-\lambda^{(1)}_{ii}-y\lambda^{(2)}_{ii}},
	\end{equation}
	which depend explicitly on $ y $. Hence, the stochastic process with such a distribution of transition times is neither a continuous-time Markov process nor an age-independent semi-Markov process.
	
	\noi For inference purposes, one may consider a parametric family $ \Lambda(y) $ given by
	\begin{equation*}
		\Lambda(y)=\Lambda^{(1)}+\Lambda^{(2)}y+\Lambda^{(3)}y^2+\dots + \Lambda^{(n+1)}y^n,\ y>0,
	\end{equation*}
	where each $ \Lambda^{(i)} $ is a rate matrix of order $ k $ and taken as a parameter. In other words, one may estimate the transition rate function with polynomials of fixed degree. In such a consideration, the number of undetermined independent parameters would be $ (n+1)(k^2-k) $. We emphasise that this family includes all Markov processes with $ k $ states and all age-independent semi-Markov processes with $ k $ states whose hazard rates are polynomials of degree not more than $ n$. Of course one may consider 
	\begin{equation*}
		\Lambda(y)=\sum\limits_{i=1}^{n+1}\Lambda^{(i)}\theta_i(y),
	\end{equation*}
	where $ \{ \theta_i \}_{i=1}^{n+1} $ is any complete orthonormal basis of $ L^2([0,\infty)) $.

	\section{Motivation for studying semi-Markov modulated markets}
	\noi In a financial market, there are numerous assets whose dynamics can be modelled by stochastic differential equations (SDEs). The drift and volatility parameters appear to be non constant when verified by empirical data. We aim, in this project, to consider a market model in which these parameters are driven by a class of pure jump processes. In the literature available on this subject, such models are referred to as regime-switching models. Although Markov switching has been better studied in the literature, we, here, aim to consider a larger class of regime switching, viz. ``age-dependent processes''. In this section, we further clarify the importance of such considerations.

	\noi The difference between markets with Markov-switching and those with semi-Markov-switching is more than superficial. To illustrate the greater applicability of the semi-Markov or age-dependent models, consider a market having only two possible regimes modulated by a semi-Markov process with two states 1 and 2, say. Let $F_i$ and $m_i$ denote the c.d.f. and mean of holding time at regime $i$ respectively for each $i$. Further assume that there is a $\delta >0$ such that $F_1(\delta)=F_2(\delta)=0$. Now consider a event $A$ in which a transition takes place at $T-\delta$, where $T$ is the expiry. Then of course there would be no more transition before expiry with probability 1. Thus all the no-arbitrage prices of European call option at time $T-\delta$ are equal to the price suggested by the Black-Scholes-Merton model with fixed parameters of that regime. On the other hand if the regimes of this real market should be modelled by a Markov process whose holding times have means $m_1$ and $m_2$ respectively, then the q-matrix would be
	$\left(
	\begin{array}{cc}
	-\frac{1}{m_1} & \frac{1}{m_1}\\
	\frac{1}{m_2} & -\frac{1}{m_2} \\
	\end{array}
	\right)$. It is evident that under this Markov switching model the conditional probability of further transition before the expiry, given the event $A$, is nonzero. Hence, the locally risk minimizing price of European call option at time $T-\delta$ should be different from Black-Scholes-Merton price with fixed parameters of that regime.
	
	\noi Such a model may, in some cases, be a better approximation to the real markets than the Markov-switching model. This provides the motivation for studying the pricing problem in a semi-Markov modulated market.
	

	\chapter{A non-local parabolic PDE}

	
	We consider a partial differential equation that arises in the derivative pricing problem in a market with semi-Markovian regime switching. This is a generalization of the Black-Scholes PDE. Market parameters are seldom constant in reality. Instead, the markets go through various phases or ``regimes'', in which each market parameter is more or less constant. We often hear of ``bull'' markets, ``flat'' markets and ``bear'' markets. Also known are low/high interest rate regimes and tight liquidity situations, etc. These can be better modelled by regime-switching models, such as those analysed in \cite{BAS}, \cite{BE}, \cite{DES}, \cite{DKR}, \cite{ECS}, \cite{GZ}, \cite{JH}, \cite{RR} and \cite{MR}. Various models of regime-switching have been studied. Work has been done on the pricing problem in a Markov-modulated market, for example, in \cite{tamal}. However, the memoryless property of Markov processes imposes certain restrictions on the model. A semi-Markov regime-switching model allows for greater flexibility and accommodates the impact of business cycles which exhibit duration dependence. In this chapter, we consider the PDE arising from an age-dependent regime-switching model, and show that this PDE is, in fact, equivalent to an equation known as a Volterra equation of the second kind. And thus, we establish the existence of a unique classical solution in an appropriate class of functions. The connection between the PDE and the pricing problem is deferred to the next chapter.

	\noi Let $\mathcal{X}:=\{1,2,\ldots,k\}$ be a finite set.
	We define the following functions
	\begin{equation}\label{market parameters1}
	r:\chi\rightarrow(0,\infty),\ \mu:(0,\infty)\times\chi\rightarrow(0,\infty),\  \sigma:(0,\infty)\times\chi\rightarrow(0,\infty).
	\end{equation}
	with $ r(i)\geq 0,~\sigma(t,i)> 0 $ for all $i\in\chi,~t\in[0,\infty)$.
	 We consider a differentiable function $\la:\mathcal{X}\times \mathcal{X}\times [0,\infty)\to [0,\infty)$ satisfying equation \eqref{lambda condition 2}, and $\lambda_{ii}(y):=-\sum_{j\neq i} \lambda_{ij}(y)$. 
	
	\noi The system of differential equations, under  consideration is given by
	\begin{eqnarray}\label{p1}
	\no \frac{\partial}{\partial t} \varphi(t, s, i, y)+ \frac{\partial} {\partial y} \varphi(t, s, i, y) + r( i ) s \frac{\partial} {\partial s} \varphi(t, s, i, y) + \frac{1}{2} \si^2( t,i ) s^2 \frac{\partial^2} {\partial s^2} \varphi(t, s, i, y)\\
	+\sum_{j\neq i}\la_{ij}(y)\big(\varphi(t, s, j, 0) -\varphi(t, s, i, y)\big) = r( i ) ~ \varphi(t, s, i, y),
	\end{eqnarray}
	\noi defined on
	\begin{equation}\label{D}
	\mathcal{D}:= \{ (t, s, i, y)\in (0,T)\times (0,\infty)\times
	\mathcal{X}\times (0,T) \mid y \in (0,t)\},
	\end{equation}
	\noi and with conditions
	\begin{align}
	\varphi (T, s, i, y)=& K(s); \tab s\in (0,\infty); \tab 0 \le
	y\le T ; \tab i= 1, 2, \ldots, k \label{boundary}
	\end{align}
	\noi where $K$ is a non-negative function of at most linear growth. This assumption on $K(s)$ is justified since we shall be considering in the next chapters defaultable bonds, which can be written as contingent claims satisfying this condition. Some of the special cases of this equation appear in \cite{DKR}, \cite{MR}, \cite{RR}, \cite{DES}, \cite{AGMKG} and \cite{BAS} for pricing a European contingent claim under certain regime switching market assumptions. Owing to the simplicity of the special case, generally authors refer to some standard results in the theory of parabolic PDE for existence and uniqueness issues. But in its general form which arises in this chapter, no such ready reference is available. So, we produce a self contained proof using Banach fixed point theorem. We accomplish this in two steps. First we consider a Volterra integral equation of second kind and establish existence and uniqueness result of that. Then we show in a couple of propositions, that the PDE and the IE problems are ``equivalent". Thus we obtain the existence and uniqueness of the PDE in Theorem \ref{theo1}. Some further properties, viz. the positivity and growth property are also obtained. It is also shown here that the partial derivative of the solution constitutes the optimal hedging strategy of the corresponding claim. We further show that the partial derivative of $\varphi$, can be written as an integration involving $\varphi$ which enables one to develop a robust numerical scheme to compute the Greeks. This study paves the way for addressing many other interesting problems involving this new set of PDEs.
	
	\section{Existence}
	Consider the following initial value problem which is known as B-S-M PDE for each $i$
	\begin{equation}\label{eq2}
	\frac{\partial\rho_{i}(t,s)}{\partial t} + r(i)s\frac{\partial\rho_{i}(t,s)}{\partial s}+\frac{1}{2}\sigma^2(t,i)s^2 \frac{\partial^2\rho_{i}(t,s)}{\partial s^2} = r(i)\rho_{i}(t,s)
	\end{equation}
	for $(t,s)\in (0,T)\times (0,\infty)$ and $\rho_{i}(T,s)=K(s)$. Here, $K$ is assumed to be a non-negative function of at most linear growth. This has a unique classical solution with at most linear growth (see \cite[pg. 202]{KK}). 
	
	\noi We define a function $L:[0,\infty)\times(0,\infty)\times(0,\infty)\times\chi\times(0,\infty)$, where
	\begin{equation}\label{ldef}
		L(t,x,s,i,v):=\frac{\ln\left(\frac{x}{s}\right)-\int_{t}^{t+v}\left(r(i)-\frac{\sigma^2(u,i)}{2}\right)\,du}{\sqrt{\int_{t}^{t+v}\sigma^2(u,i)\,du}}.
	\end{equation}
	
	\noi We also define a function
	\begin{equation}\label{alphadef}
	\alpha(x;t,s,i,v):=\frac{e^{-\frac{1}{2}L^2}}{\sqrt{2\pi}x\sqrt{\int_{t}^{t+v}\sigma^2(u,i)\,du}}.
	\end{equation}
	For notational convenience, we let $\bar{\sigma}$ denote the quantity $ \sqrt{\int_{t}^{t+v}\sigma^2(u,i)\,du} $.
	
	\begin{prop}
		The function $\alpha$ is a log-normal probability density function.
	\end{prop}
	\begin{proof}
		We at once recognise $ \alpha(x;t,s,i,v) $ to be a log-normal density function with the mean of the underlying normal distribution being $ \ln(s)+\int_{t}^{t+v}\left(r(i)-\frac{\sigma^2(u,i)}{2}\right)\,du $ and the corresponding variance being $ \int_{t}^{t+v}\sigma^2(u,i)\,du $. 
		\qed
	\end{proof}
	
	
	\begin{prop}
		\begin{equation}\label{10ae2}
		L\frac{\partial L}{\partial v}+r(i)\frac{L}{\bar{\sigma}}+\frac{\sigma^2(t+v,i)L^2}{2\bar{\sigma}^2}-\frac{\sigma^2 (t+v,i)L}{2\bar{\sigma}}=0.
		\end{equation}
	\end{prop}
	\begin{proof}
		We differentiate $L(t,x,s,i,v)$ w.r.t $v$ and apply Leibnitz's rule to get the result.
%
		 \qed
	\end{proof}
		
	\noi Set $ 	\mathcal{B}:=\left\{\varphi:\bar{\mathcal{D}}\rightarrow [0,\infty), \mathrm{continuous} \mid \|\varphi\|:=\sup_{\bar{\mathcal{D}}}\mid\frac{\varphi(t,s,i,y)}{1+s}\mid < \infty\right\} $.
	\begin{lem}\label{lm2}
		Consider the following integral equation
		\begin{align}
		\no\varphi(t,s,i,y)=&\frac{1- F(T-t+y\mid i)}{1-F(y\mid i)} \rho_{i}(t,s)+\int_0^{T-t} e^{-r(i)v} \frac{f(y+v\mid i)} {1-F(y\mid i)} \times\\
		& \sum_{j\neq i} p_{i j}(y+v) \int_0^{\infty} \varphi(t+v,x,j,0)\alpha(x;t,s,i,v)\,dx\,dv \label{pricing pde}
		\end{align}
		Then (i) the problem \eqref{pricing pde} has unique solution in $\mathcal{B}$, (ii) the solution of the integral equation is in $C^{1,2,1}(\mathcal{D})$, and (iii) $ \vf(t,s,i,y) $ is non-negative.
	\end{lem}
	\begin{proof} (i) We first note that a solution of \eqref{pricing pde} is a fixed point of the operator $A$ and vice versa, where
	\begin{eqnarray}
	\no A \varphi(t,s,i,y)&:=& \frac{1- F(T-t+y\mid i)}{1-F(y\mid i)}\rho_{i}(t,s)+\int_0^{T-t}e^{-r(i) v} \frac{f(y+v\mid i)}{1-F(y\mid i)}
	\sum_{j\neq i} p_{ij}(y+v)\\
	\no && \int^\infty _0 \varphi(t+v,x,j,0) \alpha(x;t,s,i,v)\,dx \,dv.
	\end{eqnarray}
	
	\noi  It is easy to check that for each $\vf\in \mathcal{B}$, $A\vf: \bar{\mathcal{D}}\to (0,\infty)$ is continuous. The continuity of $A\vf$ follows from that of $ \rho_i $. 

	\noi To prove that $A$ is a contraction in $\mathcal{B}$, we need to show that for $\vf_1,\vf_2\in \mathcal{B}$, $||A\varphi_1-A\varphi_2|| \leq J||\varphi_1-\varphi_2||$ where $J<1$. In order to show existence and uniqueness in the prescribed class, it is sufficient to show that $A$ is a contraction in $\mathcal{B}$. The Banach fixed point theorem ensures existence and uniqueness of the fixed point in $\mathcal{B}$. To show that for $\vf_1,\vf_2\in \mathcal{B}$, $||A\varphi_1-A\varphi_2|| \leq J||\varphi_1-\varphi_2||$ where $J<1$, we compute
	\begin{align*}
	\|A\varphi_1-A\varphi_2\|=&\sup_{\bar{\mathcal{D}}}\bigg|\frac{ A\varphi_1-A\varphi_2}{1+s}\bigg|\\
	=&\sup_{\bar{\mathcal{D}}}\bigg|\int^{T-t}_0 e^{-r(i) v} \frac{f(y+v\mid i)}{1-F(y\mid i)} \sum_{j\neq i} p_{ij}(y+v)\times\\
	& \int^\infty_0 (\varphi_1-\varphi_2)(t+v,x,j,0)\frac{\alpha(x;t,s,i,v)}{1+s}dx dv\bigg|\\
	\leq& \sup_{\bar{\mathcal{D}}} \bigg| \int^{T-t}_0 e^{-r(i) v} \frac{f(y+v\mid i)}{1-F(y\mid i)}\sum_{j\neq i} p_{ij}(y+v) \int^\infty_0 (1+x)\times\\
	&\sup_{(t',x',j',y')\in\bar{\mathcal{D}}}\bigg|\frac{\varphi_1(t',x',j',y')-\varphi_2(t',x',j',y')}{1+x'}\bigg| \frac{\alpha(x;t,s,i,v)}{1+s}dx dv\bigg|\\
	=& \sup_{\bar{\mathcal{D}}} \bigg| \int^{T-t}_0 e^{-r(i) v} \frac{f(y+v \mid i)}{1-F(y \mid i)}\|\varphi_1- \varphi_2 \| \frac{a(t,s)}{1+s}dv\bigg|
	\end{align*}
	where, 
	\begin{align*}
		a(t,x,s,i,v):=& \int^\infty_0 (1+x) \alpha(x;t,s,i,v) dx\\
		 =& 1 + \exp\left \lbrace\ln s +  r(i)-\int_{t}^{t+v}\frac{\sigma^2(u,i)}{2}\,du+\int_{t}^{t+v}\frac{\sigma^2(u,i)}{2}\,du\right\rbrace\\
		  =& 1+ s e^{r(i)v}.
	\end{align*}
	
	\noi Thus, $\|A\varphi_1-A\varphi_2\| \, \leq\, J\|\varphi_1-\varphi_2\|$ where,
	\begin{eqnarray*}
	J&=&\sup_{\bar{\mathcal{D}}} \bigg|\int^{T-t}_0e^{-r(i) v}\frac{f(y+v \mid i)}{1-F(y \mid i)}\frac{1+se^{r(i)v}}{1+s}dv\bigg|\\
	&\leq &\sup_{\bar{\mathcal{D}}}\bigg( \frac{1}{1-F(y\mid i)}\int^{T-t}_0 f(y+v|i)dv\bigg)\\
	&=&\sup_{\bar{\mathcal{D}}}\bigg(\frac{F\left( y+T-t\mid i\right) -F(y|i)}{1-F\left(y|i\right)}\bigg)\\
	&<&\frac{1-F(y|i)}{1-F(y|i)}~=~1
	\end{eqnarray*}
	using $r(i)\geq 0$ and the properties of $\lambda$ and $F$.
	
	\noi (ii) Using equation\eqref{lambda condition 2} and smoothness of $\rho_i$ for each $i$, the first term on the right hand side is in $C^{1,2,1}(\mathcal{D})$. Under the assumptions on $\lambda$ and $F$, the second
	term is continuous differentiable in $y$ and twice continuously differentiable in s, follows immediately. The continuous
	differentiability in $t$ follows from the fact that the term $\varphi(t+v,x,j,0)$
	is multiplied by $C^1((0,\infty))$ functions in $v$ and then integrated over $v\in(0,T-t)$.
	Hence $\varphi(t,s,i,y)$ is in $C^{1,2,1}(\mathcal{D})$.

	\noi (iii) We have shown that $ A:\mathcal{B}\rightarrow\mathcal{B} $ is a contraction. It is evident that equation \eqref{eq2} has a non-negative solution. Since all coefficients of the integral equation \eqref{pricing pde} are non-negative, $ A\vf\geq 0 $ for $ \vf\geq 0 $. Now let $ V:=\{\phi\in\mathcal{B}\mid \phi \geq 0 \} $. Then, $ V $ is a closed subset of $ \mathcal{B} $. Consider $ A:V\rightarrow V $, and let $ v_0\in V $. Define a sequence $ \{v_n\}_{n\geq 0} $, such that $ v_n:=A^n v_0 $. Then $ v_n\in V $. We note that
	\begin{align*}
		\| v_{m+p}-v_m \|=&\| (A^p-I)A^m v_0 \|\\
		\leq & \| (A^p-I) \|.\|A\|^m.\|v_0\|.
	\end{align*}
	We have shown that $ \|A\varphi_1-A\varphi_2\| \, \leq\, J\|\varphi_1-\varphi_2\| $, where $ J<1 $. Hence $ \|A\|<1 $, which means $ \| v_{m+p}-v_m \|\rightarrow 0 $ as $ m\rightarrow \infty $. Thus, $ \{v_n\}_{n\geq 0} $ is a Cauchy sequence. Since $ V $ is closed, $ v_n\rightarrow v $, where $ v\in V $. The continuity of $ A $ implies that $ Av_n\rightarrow Av $. Also, $ Av_n=v_{n+1}\rightarrow v $. This means $ Av=v $, i.e. $ v $ is a fixed point of $ A $. 
	
	\noi We have already shown that $ A $ has a fixed point in $ \mathcal{B} $. This fixed point is $ v $, which is an element of $ V $. In other words, $ v $ is non-negative. Thus, we have established that the fixed point of $ A $ in $ \mathcal{B} $ is non-negative, i.e. $ \vf $ is non-negative.\qed
	\end{proof}
	
	\begin{lem}\label{dirac}
		Let $ \varphi $ be the solution of equation \eqref{pricing pde}. Then
		\begin{equation*}
			\lim_{u\downarrow 0}\int^\infty_0 \varphi(t+u,x,j,0)\alpha(x;t,s,i,u)\,dx=\varphi(t,s,j,0).
		\end{equation*}
	\end{lem}
	\begin{proof}
		Since $ \vf(t,\cdotp,i,y) $ is of at most linear growth, there exist positive constants $ k_1 $ and $ k_2 $ such that $ \vf(t,s,i,y)\leq k_1 + k_2s $ for all $ s $. Let $ \{u_l\}_{l\in\mathbb{N}} $ be a decreasing sequence on $ (0,1) $ such that $ u_l\rightarrow 0 $. Let $ \alpha_l(x):=\alpha(x;t,s,i,u_l) $. Since $ \alpha_l $ is a lognormal density function for each $ l $, the sequence $ \{ \alpha_l \}_{l\in\mathbb{N}} $ is uniformly integrable, that is
		\begin{equation*}
			\lim_{k\rightarrow \infty}\sup_l\int_k^{\infty}x\alpha_l(x)\,dx=0.
		\end{equation*}
		We fix $ t $ and $ s $. Thus, for any $ \epsilon>0 $, we can find $ K>0 $ such that $ \int_K^{\infty}(k_1+k_2x)\alpha_l(x)\,dx<\epsilon $ for all $ l\in\mathbb{N} $. Now let $ \{\vf_n\}_{n\in\mathbb{N}} $ be a non-negative increasing sequence of step functions in $ x $ converging to $ \vf $ pointwise. Then, given $ \epsilon>0 $ and $ K $, we can find $ N $ such that for all $ n\geq N $,
		\begin{align*}
			\lefteqn{\int_0^{\infty}(\vf(t+u_l,x,j,0)-\vf_n(t+u_l,x,j,0))\alpha_l(x)\,dx}\\
			=&\int_0^{K}(\vf(t+u_l,x,j,0)-\vf_n(t+u_l,x,j,0))\alpha_l(x)\,dx + \int_K^{\infty}(\vf(t+u_l,x,j,0)-\vf_n(t+u_l,x,j,0))\alpha_l(x)\,dx\\ 
			\no\leq & \epsilon\alpha_l([0,K])+\int_K^{\infty}(k_1+k_2x)\alpha_l(x)\,dx+\int_{0}^{K}\left[ \vf(t+u_l,x,j,0)-\vf(t,x,j,0) \right]\alpha_l(x)\,dx\\
			\no<& 2\epsilon+\int_{0}^{K}\left[ \vf(t+u_l,x,j,0)-\vf(t,x,j,0) \right]\alpha_l(x)\,dx,
		\end{align*}
		where $ \alpha_l(A):=\int_A \alpha_l(x)\,dx $.
		Also,
		\begin{align*}
			\int_0^{\infty}(\vf(t+u_l,x,j,0)-\vf_n(t+u_l,x,j,0))\alpha_l(x)\,dx=&\int_0^{\infty}\vf(x)\alpha_l(x)\,dx-\sum_{i=1}^{K_n}\vf_n(x_i)\alpha_l(I_i),
		\end{align*}
		where $ \vf_n(x)=\sum_{i=1}^{K_n}a_i \mathds{1}_{I_i}(x) $ and $ x_i\in I_i $. As $ l\rightarrow\infty $,
		\begin{equation*}
			\alpha_l(I_i)\rightarrow\begin{cases}
			0,\text{ if } s\notin I_i,\\
			1,\text{ if } s\in I_i.
			\end{cases}
		\end{equation*}
		Hence, for each $ n $,
		\begin{equation*}
			\lim_{l\rightarrow\infty}\int_0^{\infty}\vf_n\alpha_l(x)\,dx = \vf_n(s).
		\end{equation*}
		Thus, for $ n\geq N(\epsilon,K) $,
		\begin{align*}
			0\leq&\lim_{l\rightarrow\infty}\int_0^{\infty}\vf(t+u_l,x,j,0)\alpha_l(x)\,dx-\vf_n(s)\\
			\leq& 2\epsilon+\lim_{l\rightarrow\infty}\int_{0}^{K}\left[ \vf(t+u_l,x,j,0)-\vf(t,x,j,0) \right]\alpha_l(x)\,dx\\
			=& 2\epsilon,
		\end{align*}
		since $ \vf(\cdotp,s,i,y) $ is smooth.	Thus, $ \lim_{n\rightarrow\infty}\vf_n(t,s,j,0) = \lim_{l\rightarrow\infty}\int_0^{\infty}\vf(t+u_l,x,j,0)\alpha_l(x)\,dx $. 
		\qed
	\end{proof}

	\begin{prop}\label{theo3} The unique solution of \eqref{pricing pde} also solves the initial value problem \eqref{p1}-\eqref{boundary}.
	\end{prop}
	\proof
	Let $\varphi$ be the solutions of \eqref{pricing pde}. Thus using \eqref{pricing pde}, $\varphi(T,s,i,y)=\rho_{i}(T,s)=K(s)$, i.e., the condition \eqref{boundary} holds. From Lemma \ref{lm2} (ii), $\varphi$ is in $C^{1,2,1}(\mathcal{D})$. Hence we can perform the partial differentiations w.r.t. $t$ and $y$ on the both sides of \eqref{pricing pde}. We obtain
	\begin{align}\label{2ae}
	\no \frac{\partial}{\partial t} \varphi(t, s, i, y)=&\frac{f(T-t+y|i)}{1-F(y\mid i)}\rho_{i}(t,s)+\frac{ 1- F(T-t+y\mid i)}{\left(1-F(y\mid i)\right)}\frac{\partial\rho_{i}(t,s)}{\partial t}- e^{-r(i)(T-t)}\frac{f(y+T-t\mid i)}{1-F(y\mid i)}\times\\
	&\no\sum_{j\neq i} p_{ij}(y+T-t)\int_0^\infty \varphi(T,x,j,0)\alpha(x;t,s,i,T-t)dx+\int^{T-t}_0 e^{-r(i)v}\frac{f(y+v\mid i)}{1-F(y\mid i)}\times \\
	&\no\sum_{j\neq i} p_{ij}(y+v)\int^\infty_0 \frac{\partial \varphi}{\partial t}(t+v,x,j,0) \alpha(x;t,s,i,v) dx dv\\
	&\no + \int_0^{T-t} e^{-r(i)v} \frac{f(y+v\mid i)} {1-F(y\mid i)} \times\sum_{j\neq i} p_{i j}(y+v)\times\\
	& \int_0^{\infty} \varphi(t+v,x,j,0)\alpha(x;t,s,i,v)\left(\frac{\sigma^2(t+v,i)-\sigma^2(t,i)}{2} \right) \left( \frac{L^2}{\bar{\sigma}^2}-\frac{L}{\bar{\sigma}}-\frac{1}{\bar{\sigma}^2} \right) \,dx\,dv
	\end{align}
	by differentiating w.r.t. $t$ under the sign of integral. Now, before we take the partial derivative w.r.t. $y$ on both sides of \eqref{pricing pde}, we first simplify the right-hand side. Let $q_{ij}(y+v):=f(y+v\mid i)p_{ij}(y+v)$. Then
	\begin{align*}
		\no \frac{\partial}{\partial y} \varphi(t, s, i, y)=&\no -\frac{f(T-t+y \mid i)}{1-F(y\mid i)}\rho_{i}(t,s)+\frac{1- F(T-t+y\mid i)}{\left(1-F(y\mid i)\right)^2}f(y|i)\rho_{i}(t,s)\\
		&+\frac{\partial}{\partial y}\int_0^{T-t} e^{-r(i)v} \frac{q_{ij}(y+v)} {1-F(y\mid i)} \int_0^{\infty} \varphi(t+v,x,j,0)\alpha(x;t,s,i,v)\,dx\,dv.
	\end{align*}
	The last term can be simplified further. 
	\begin{align*}
		&\frac{\partial}{\partial y}\int_0^{T-t} e^{-r(i)v} \frac{f(y+v\mid i)} {1-F(y\mid i)}\sum_{j\neq i} p_{i j}(y+v) \int_0^{\infty} \varphi(t+v,x,j,0)\alpha(x;t,s,i,v)\,dx\,dv \\
		=&\sum_{j\neq i}\frac{\partial}{\partial y}\left[ \dfrac{1}{1-F(y\mid i)}\int_0^{T-t} \left( e^{-r(i)v}\int_0^{\infty} \varphi(t+v,x,j,0)\alpha(x;t,s,i,v)\,dx\right) q_{ij}(y+v)\,dv  \right] 
	\end{align*}
	Let $ b_{ij}(v;t,x,s):=e^{-r(i)v}\int_0^{\infty} \varphi(t+v,x,j,0)\alpha(x;t,s,i,v)\,dx $. Also let $ \tilde{q}_{ij}(y):=\int_{0}^{y}q_{ij}(w)\,dw $, so that $ \tilde{q}_{ij}'(y)=q_{ij}(y) $. Then, using the integration by parts formula, we get
	\begin{align*}
		\int_0^{T-t}b_{ij}(v;t,x,s)q_{ij}(y+v)\,dv =&\left[ b_{ij}(v;t,x,s) \tilde{q}_{ij}(y+v) \right]^{T-t}_0\\
		& -  \int_0^{T-t}\frac{\partial{b_{ij}(v;t,x,s)}}{\partial v}\tilde{q}_{ij}(y+v)\,dv.
	\end{align*}
	Now, 
	\begin{align*}
		b_{ij}(T-t;t,x,s)\tilde{q}_{ij}(y+T-t) =&  \no  e^{-r(i)(T-t)}\tilde{q}_{ij}(y+T-t)\int^\infty_0 \varphi(T,x,j,0)\alpha(x;t,s,i,T-t)\,dx
	\end{align*}
	while
	\begin{align*}
		b_{ij}(0;t,x,s)\tilde{q}_{ij}(y)=& \tilde{q}_{ij}(y)\left[ \lim_{u\downarrow 0}\int^\infty_0 \varphi(t+u,x,j,0)\alpha(x;t,s,i,u)\,dx\right] \\
		=& \tilde{q}_{ij}(y)\varphi(t,s,j,0)
	\end{align*}
	by lemma \ref{dirac}.
	
	\noi Hence, the partial derivative of $\varphi$ w.r.t $y$ is
	\begin{align}\label{3ae}
	\no \frac{\partial}{\partial y} \varphi(t, s, i, y)=&\no -\frac{f(T-t+y \mid i)}{1-F(y\mid i)}\rho_{i}(t,s)+\frac{1- F(T-t+y\mid i)}{\left(1-F(y\mid i)\right)^2}f(y|i)\rho_{i}(t,s)+\frac{f(y|i)}{1-F(y\mid i)}\times\\
	&\no \Big(\varphi(t,s,i,y)- \frac{1- F(T-t+y \mid i)}{1-F(y\mid i)}\rho_{i}(t,s)\Big)+e^{-r(i)(T-t)}\frac{f(T-t+y\mid i)}{1-F(y\mid i)}\times \\
	&\no\sum_{j\neq i} p_{ij}(y+T-t)\int^\infty_0 \varphi(T,x,j,0)\alpha(x;t,s,i,T-t)\,dx\\
	&\no -\frac{f(y\mid i)}{1-F(y\mid i)}\sum_{j\neq i} p_{ij}(y)\varphi(t,s,j,0)\\
	&\no-\int^{T-t}_0 e^{-r(i) v} \frac{f(y+v\mid i)}{1-F(y\mid i)}\int^\infty_0 \alpha(x;t,s,i,v) \Bigg\{-r(i) \sum_{j\neq i} p_{ij}(y+v) \varphi(t+v,x,j,0)\\
	&\no-\sum p_{ij}(y+v) \varphi(t+v,x,j,0)\left(L \frac{\partial L}{\partial v}+\frac{\sigma^2(t+v,i)}{2\bar{\sigma}^2}\right)\\
	&+\sum_{j\neq i} p_{ij}(y+v) \frac{\partial{\varphi(t+v,x,j,0)}}{\partial t}\Bigg\}\,dx\,dv.
	\end{align}
	\noi By adding equations \eqref{2ae} and \eqref{3ae}, we get
	\begin{align}\label{6ae}
	\no\lefteqn{\frac{\partial}{\partial t} \varphi(t, s, i, y)+\frac{\partial} {\partial y} \varphi(t, s, i, y)}\\
	=&\no\frac{ 1- F(T-t+y\mid i)}{1-F(y\mid i)}\frac{\partial\rho_{i}(t,s)}{\partial t}+ \frac{f(y|i)}{1-F(y\mid i)}\left(
	\varphi(t,s,i,y)-\sum_{j\neq i} p_{ij}(y)\varphi(t,s,j,0)\right)\\
	 &+ \int^{T-t}_0 e^{-r(i)v}\no \frac{f(y+v|i)}{1-F(y|i)}\sum_{j\neq i} p_{ij}(y+v) \int_0^\infty \varphi(t+v,x,j,0)\alpha(x;t,s,i,v)\times\\
	&\left(r(i)+L\frac{\partial L}{\partial v}+\frac{\sigma^2(t+v,i)L^2}{2\bar{\sigma}^2}-\frac{\sigma^2(t,i)L^2}{2\bar{\sigma}^2}-\frac{\sigma^2(t+v,i)L}{2\bar{\sigma}}+\frac{\sigma^2(t,i)L}{2\bar{\sigma}}+\frac{\sigma^2(t,i)}{2\bar{\sigma}^2}\right)\,dx\,dv.
	\end{align}
	\noi Now we differentiate both sides of \eqref{pricing pde} w.r.t. $s$ once and twice respectively and obtain
	\begin{eqnarray}\label{4ae}
	\no \frac{\partial}{\partial s} \varphi(t, s, i, y)&=& \frac{ 1- F(T-t+y\mid i)}{1-F(y\mid i)}\frac{\partial\rho_{i}(t,s)}{\partial s}+\int_0^{T-t} e^{-r(i)v}\frac{f(y+v\mid i)} {1-F(y\mid i)} \sum_{j\neq i} p_{i j}(y+v)\times\\
	&&\int_0^{\infty} \varphi(t+v,x,j,0)\alpha(x;t,s,i,v)\frac{L}{s\bar{\sigma}}\,dx\,dv,
	\end{eqnarray}
	\begin{eqnarray}\label{5ae}
	\no \frac{\partial^2} {\partial s^2} \varphi(t, s, i, y)&=& \frac{ 1- F(T-t+y\mid i)}{1-F(y\mid i)}\frac{\partial^2 \rho_{i}(t,s)}{\partial s^2}+\int_0^{T-t} e^{-r(i)v}\frac{f(y+v\mid i)} {1-F(y\mid i)} \sum_{j\neq i} p_{i j} (y+v)\times\\
	&&\int_0^{\infty} \varphi(t+v,x,j,0)\alpha(x;t,s,i,v) \frac{1}{s^2}\left(\frac{L^2}{\bar{\sigma}^2} - \frac{L}{\bar{\sigma}} -\frac{1}{\bar{\sigma}^2}\right)\,dx\,dv.
	\end{eqnarray}
	\noi From equations \eqref{4ae} and \eqref{5ae}, we get
	\begin{align}\label{7ae}
	\no \lefteqn{r(i) s \frac{\partial \varphi}{\partial s}+ \frac{1}{2}\sigma^2(i) s^2 \frac{\partial^2 \varphi}{\partial s^2}}\\
	=&\no \frac{ 1- F(T-t+y\mid i)}{1-F(y\mid i)}\left(r(i)s\frac{\partial\rho_{i}(t,s)}{\partial s}+\frac{1}{2}\sigma^2(i)s^2 \frac{\partial^2\rho_{i}(t,s)}{\partial s^2}\right)+\int_0^{T-t}e^{-r(i)v}\frac{f(y+v\mid i)}{1-F(y|i)}\times\\
	&\sum_{j\neq i} p_{ij}(y+v)\int_0^\infty \varphi(t+v,x,j,0)\alpha(x;t,s,i,v)\left(\frac{r(i)L}{\bar{\sigma}} +\frac{\sigma^2(t,i)L^2}{2\bar{\sigma}^2}-\frac{\sigma^2(t,i)L}{2\bar{\sigma}}-\frac{\sigma^2(t,i)}{2\bar{\sigma}^2}\right)\,dx\,dv.
	\end{align}
	\noi Finally, from equations \eqref{pricing pde}, \eqref{eq2}, \eqref{10ae2}, \eqref{6ae} and \eqref{7ae} we get
	\begin{eqnarray*}\label{8ae}
	\lefteqn{\no\frac{\partial}{\partial t} \varphi(t, s, i, y)+\frac{\partial} {\partial y} \varphi(t, s, i, y)+r(i) s \frac{\partial }{\partial s}\varphi(t, s, i, y)+ \frac{1}{2}\sigma(t,i)^2(i) s^2\frac{\partial^2 }{\partial s^2}\varphi(t, s, i, y)}\\
	&=&\no \frac{ 1- F(T-t+y\mid i)}{1-F(y\mid i)}\left[\frac{\partial\rho_{i}(t,s)}{\partial t} + r(i)s\frac{\partial\rho_{i}(t,s)}{\partial s}+\frac{1}{2}\sigma^2(t,i)s^2 \frac{\partial^2\rho_{i}(t,s)}{\partial s^2}\right] -\frac{f(y\mid i)}{1-F(y\mid i)}\times \\
	&&\sum_{j\neq i} p_{ij}(y)(\varphi(t,s,j,0)-\varphi(t,s,i,y)) + r(i) \left(\varphi(t,s,i,y)- \frac{ 1- F(T-t+y\mid i)}{1-F(y\mid i)} \rho_{i}(t,s)\right)\\
	&=& -\frac{f(y\mid i)}{1-F(y\mid i)}\sum_{j\neq i} p_{ij}(y)(\varphi(t,s,j,0)-\varphi(t,s,i,y))+r(i)\varphi(t,s,i,y).
	\end{eqnarray*}
	\noi Thus equation \eqref{p1} holds. \qed

	\noi From Lemma \ref{lm2} and Proposition \ref{theo3} it follows that \eqref{p1}-\eqref{boundary} has a classical solution. We prove uniqueness in the following section.

	\section{Uniqueness}\label{uniqueness}
	
	We consider equations \eqref{p1}-\eqref{boundary}.

	\noi It is interesting to note that although the domain $ \mathcal{D} $ has non-empty boundary, we have obtained existence of a unique solution of the IE in the prescribed class without imposing boundary conditions. Furthermore, we shall show that the uniqueness of the IE implies uniqueness of the PDE. This invokes an immediate surprise as we know that boundary condition is important for uniqueness for a non-degenerate parabolic PDE. In this connection, we would like to recall, here the PDE is degenerate. For one part of boundary, i.e $ s=0 $, coefficients of all the differential operators w.r.t. $ s $ vanish. Thus, it is natural to expect that a condition on $ s=0 $ might not be needed for uniqueness. In other words, the PDE would have non-existence for any boundary condition except possibly only an appropriate one obtained from the terminal condition. We further clarify this apparently vague reasoning with a precise calculation below. Other than $ s=0 $, the remaining parts of the boundary is due to the boundary of the $ y $ variable, i.e $ y=0 $ and $ y=t $. Here the non-rectangular nature of $ \mathcal{D} $ becomes apparent. We recall that we address a terminal value problem, thus the range of $ y $ shrinks linearly in $ t $ as $ t $ decreases to zero. On the other hand only the first order differential operator w.r.t. $ y $ appears in the PDE. Thus the absence of boundary data is not leading to a under-determined problem.
	
	\noi We consider continuous solutions to the problem  \eqref{p1}-\eqref{boundary} on the closure of the domain $ \mathcal{D} $, in particular, the set $ \{(t,s,i,y)\in\bar{\mathcal{D}}\mid s=0\} $. For $s=0$, the PDE is
	\begin{equation}\label{characteristics}
		\left( \frac{\partial }{\partial t}+\frac{\partial }{\partial y} \right) \vf(t,0,i,y)+\sum_{j\neq i}\lambda_{ij}(y)[\vf(t,0,j,0)-\vf(t,0,i,y)]=r_i\vf(t,0,i,y).
	\end{equation}
	Let $\hat{\vf}_i(t,y):=\vf(t,0,i,y)$. Then,
	\begin{equation*}
		\left( \frac{\partial }{\partial t}+\frac{\partial }{\partial y} \right)\hat{\vf}_i(t,y)+\sum_{j\neq i}\lambda_{ij}(y)[\hat{\vf}_j(t,0)-\hat{\vf}_i(t,y)]=r_i\hat{\vf}_i(t,y),
	\end{equation*}
	with the terminal condition $ \hat{\vf}_i(T,y)=K(0) $.
	Now, for any $ t_0<T $, consider $c_{t_0}(t):=t-t_0$. Then,
	\begin{equation*}
		\frac{d}{dt}\hat{\vf}_i(t,c_{t_0}(t))=\left( \frac{\partial }{\partial t}+\frac{\partial }{\partial y} \right)\hat{\vf}_i(t,c_{t_0}(t)).
	\end{equation*}
	Let $ g_i(t;t_0):=\hat{\vf}_i(t,c_{t_0}(t)) $. Then
	\begin{equation*}
		\frac{d}{dt}g_i(t;t_0)+\sum_{j\neq i}\lambda_{ij}(c_{t_0}(t))[\hat{\vf}_j(t,0)-g_i(t)]=r_i g_i(t;t_0).
	\end{equation*}
	Hence,
	\begin{equation*}
		\frac{dg_i(t;t_0)}{dt}=p(t)g_i(t;t_0)-q(t),\quad g_i(T;t_0)=K(0)
	\end{equation*}
	where $ p(t):=r_i+\sum_{j\neq i}\lambda_{ij}(c(t)) $ and $ q(t):=\sum_{j\neq i}\lambda_{ij}(c(t))\hat{\vf}_j(t,0) $. This is a first-order linear ODE, which can easily be solved to give
	\begin{equation*}
		g_i(t;t_0)=\int_{t}^{T}e^{-\int_{t_0}^{u}(r_i+\sum_{j\neq i}\lambda_{ij}(c_{t_0}(s)))\,ds} \sum_{j\neq i}\lambda_{ij}(c_{t_0}(u))\hat{\vf}_j(u,0)\,du - K(0) e^{-\int_{t_0}^{u}(r_i+\sum_{j\neq i}\lambda_{ij}(c_{t_0}(s)))\,ds}.
	\end{equation*}
	
	\noi Now, $g_i(t_0,t_0)=\hat{\vf}_i(t_0,0)$. Thus, we obtain the following equation for $\hat{\vf}$:
	
	\begin{equation}\label{integral eqn 2}
		\hat{\vf}_i(t,0)=\int_{t}^{T}e^{-\int_{t}^{u}(r_i+\sum_{j\neq i}\lambda_{ij}(c_t(s)))\,ds} \sum_{j\neq i}\lambda_{ij}(c_t(u))\hat{\vf}_j(u,0)\,du - K(0) e^{-\int_{t}^{u}(r_i+\sum_{j\neq i}\lambda_{ij}(c_t(s)))\,ds},
	\end{equation}
	
	\noi This is an integral equation in $ \hat{\vf}(t,0) $. If we show that this system of integral equations has a unique solution, our reasoning regarding the redundancy of the boundary condition on $ s $ will be justified. To this end, we proceed in a manner similar to the proof of Lemma \ref{lm2}. We define the operator $A$ to be
	\begin{equation}
		A\hat{\vf}_i(t,0)=\int_{t}^{T}e^{-\int_{t}^{u}(r_i+\sum_{j\neq i}\lambda_{ij}(c_t(s)))\,ds} \sum_{j\neq i}\lambda_{ij}(c_t(u))\hat{\vf}_j(u,0)\,du - K(0) e^{-\int_{t}^{u}(r_i+\sum_{j\neq i}\lambda_{ij}(c_t(s)))\,ds}.
	\end{equation}
	The solution to equation \eqref{integral eqn 2} is obviously a fixed point of the operator $A$. If we are able to establish that $A$ is a contraction in the class of functions we are about to consider, Banach fixed point theorem can be used to show that the integral equation \eqref{integral eqn 2} has a unique solution which is a fixed point of $A$. We define $ \Gamma:=\chi\times[0,T] $ to be the domain which we shall now consider. Consider the Banach space $ \mathcal{B}=C\left( \Gamma \right) $, endowed with the sup-norm. In order to show that $A$ is a contraction, we need to prove that for $\hat{\vf}_1,\hat{\vf}_2\in \mathcal{B}$, $||A\hat{\varphi}^1-A\hat{\varphi}^2|| \leq J||\hat{\varphi}^1-\hat{\varphi}^2||$ where $J<1$. Now,
	\begin{align*}
		A(\hat{\vf}^1_i-\hat{\vf}^2_i)=&\int_{t}^{T}e^{-\int_{t}^{u}(r_i+\sum_{j\neq i}\lambda_{ij}(c_t(s)))\,ds} \sum_{j\neq i}\lambda_{ij}(c_t(u))\left(\hat{\vf}^1_j(u,0)-\hat{\vf}^2_j(u,0)\right)\,du\\
		\leq& \no \int_{t}^{T}e^{-\int_{t}^{u}(r_i+\sum_{j\neq i}\lambda_{ij}(c_t(s)))\,ds} \sum_{j\neq i}\lambda_{ij}(c_t(u))\sup_{u,j}\left(\hat{\vf}^1_j(u,0)-\hat{\vf}^2_j(u,0)\right)\,du.
	\end{align*}
	
	\noi Since $ r(i)>0 $ for all $i$,
	\begin{align*}
		A(\hat{\vf}^1_i-\hat{\vf}^2_i)\leq & \| \hat{\vf}^1-\hat{\vf}^2 \| \int_{t}^{T}e^{-\int_{t}^{u}(r_i+\sum_{j\neq i}\lambda_{ij}(s-t))\,ds} \sum_{j\neq i}\lambda_{ij}(u-t)\,du\\
		=& \| \hat{\vf}^1-\hat{\vf}^2 \| \int_{t}^{T}e^{-r_i(u-t)}e^{-\int_{t}^{u}\sum_{j\neq i}\lambda_{ij}(s-t)\,ds} \sum_{j\neq i}\lambda_{ij}(u-t)\,du\\
		<& \| \hat{\vf}^1-\hat{\vf}^2 \| \int_{t}^{T}e^{-\int_{t}^{u}\sum_{j\neq i}\lambda_{ij}(s-t)\,ds} \sum_{j\neq i}\lambda_{ij}(u-t)\,du\\
		=&\| \hat{\vf}^1-\hat{\vf}^2 \| \int_{t}^{T} \frac{d}{du}\left( e^{-\int_{t}^{u}\sum_{j\neq i}\lambda_{ij}(s-t)\,ds} \right) \,du\\
		=& \| \hat{\vf}^1-\hat{\vf}^2 \| \left( 1-e^{-\int_{t}^{T}\sum_{j\neq i}\lambda_{ij}(s-t)\,ds} \right) \\
		=& J\| \hat{\vf}^1-\hat{\vf}^2 \|,
	\end{align*}
	where $ J=1-e^{-\int_{t}^{T}\sum_{j\neq i}\lambda_{ij}(s-t)\,ds}<1 $. This proves that $A$ is, in fact, a contraction. Thus, the uniqueness of $\hat{\varphi}$, the solution to equation \eqref{integral eqn 2} is established. The uniqueness of $ \hat{\vf}(t_0,0) $ for all $ t_0\in[0,T] $ implies the uniqueness of $ g(t;t_0) $ for all $ t \geq t_0 \geq 0 $. Also, $ \hat{\vf}_i(t,t-t_0) $ is unique for all $ t\in[t_0,T] , ~ t_0\in[0,T]$. Since, for $y\in[0,t]$, $ \vf_i(t,0,i,y)=\hat{\vf}_i(t,y)=\hat{\vf}_i(t,t-(t-y)) $, with $t-y\in[0,t]$, equation \eqref{characteristics} has a unique solution. Hence, $\vf(t,s,i,y)$ is unique for $ s=0 $.

	\begin{prop}\label{theo4} Assume \eqref{lambda condition 1} and \eqref{lambda condition 2}. We also assume that the transition matrix $ \tilde{p}_{ij}:=\int_0^\infty p_{ij}(y)\,dF_i(y) $ is irreducible. Let $\varphi$ be a classical solution of \eqref{p1}-\eqref{boundary}. Then (i) $\varphi$ solves the integral equation \eqref{pricing pde}; (ii) $ \vf(t,s,i,y)\leq k_1+k_2s $ for some $ k_1,k_2> 0 $.
	\end{prop}
	\proof (i) Let $(\tilde\Omega,\tilde{\mathcal{F}}, \tilde P)$ be a probability space which holds a standard Brownian motion $W$ and the Poisson random measure $\wp$ independent of $W$. Let $\tilde{S}_t$ be the strong solution of the following SDE
	\begin{eqnarray*}
		d\tilde{S}_t = \tilde{S}_t(r(X_t)dt + \si(t,X_t)dW_t),\tab \tilde{S}_0>0
	\end{eqnarray*}
	where $X_t$ is the age-dependent process given by equations \eqref{xdef} and \eqref{ydef}. Let $\tilde{ \mathcal{F}}_t$ be the underlying filtration generated by $ \tilde{S}_t $ and $ X_t $ satisfying the usual hypothesis. We observe that the process $\{(\tilde{S}_t,X_t,Y_t)\}_t$ is Markov with infinitesimal generator $\mathcal{A}_t$, where 	
	\begin{align*}
		\mathcal{A}_t \varphi(s,i,y)=&\frac{\partial \vf}{\partial y}  (s, i, y) + r(i) s \frac{\partial \vf}{\partial s} (s, i, y)+ \frac{1}{2} \si^2(t,i) s^2 \frac{\partial^2 \vf} {\partial
		s^2} (s, i, y) \\
	&+ \sum_{j\neq i}\lambda_{ij}(y) \big(\varphi(s,j,0) -\varphi(s,i,y)\big)
	\end{align*}
	for every function $\varphi$ which is compactly supported $C^2$ in $s$ and $C^1$ in $y$. If $\vf$ is the classical solution of \eqref{p1}-\eqref{boundary} then by using the It\^{o}'s formula on $N_t := e^{-\int_0^t r(X_u)du} \varphi(t,\tilde{S}_t,X_t,Y_t)$, we get
	\begin{eqnarray*}
	dN_t &=& e^{-\int_0^t r(X_u)du}\left(- r(X_t) \varphi (t,\tilde{S}_t,X_t,Y_t)+ \frac{\partial \varphi}{\partial t }(t,\tilde{S}_t,X_t,Y_t) + \mathcal{A}_t \varphi (t,\tilde{S}_t,X_t,Y_t)\right)dt + dM_t
	\end{eqnarray*}
	where $M_t$ is a local martingale. Thus from \eqref{p1} and above expression, $N_t$ is also an $ \tilde{ \mathcal{F}}_t$ local martingale. The definition of $N_t$ suggests that there are constants $k_1$ and $k_2$ such that $ |N_t|\le k_1 +k_2 \tilde S_t$ for each $t$, since $\vf$ has at most linear growth. Again, from the following expression
	$$ \tilde S_t = \tilde S_0 \exp\left(\int_0^t (r(X_u)-\frac{1}{2}\sigma(u,X_u)^2)\,du + \int_0^t \sigma(u,X_u)\,dW_u\right)
	$$
	one concludes that $\tilde S_t$ is a submartingale with finite expectation. Therefore Doob's inequality can be used to obtain $ E\sup_{u\in [0,t]}|N_u| < \infty$ for each $t$. Thus $\{N_t\}_t$ is a martingale. Hence
	\begin{equation}\label{eq18}
	\varphi(t,\tilde{S}_t,X_t,Y_t)= e^{\int_0^t r(X_u)du}N_t = E[e^{\int_0^t r(X_u)du}N_T\mid \mathcal{ F}_t]= E[e^{-\int_t^T r(X_u)du} K(\tilde{S}_T)\mid \tilde{S}_t,X_t,Y_t].
	\end{equation}
	By conditioning at transition times and using the conditional lognormal distribution of $\tilde{S}_t$, we get
	\begin{align*}
		\lefteqn{\varphi(t,\tilde{S}_t,X_t,Y_t)}\\
		=& E[E[e^{-\int_t^T r(X_u)du} K(\tilde{S}_T) \mid \tilde{S}_t, X_t=i, Y_t, T_{n(t)+1}]\mid \tilde{S}_t, X_t=i, Y_t]\\
		=& P(T_{n(t)+1} > T\mid X_t,Y_t) E[e^{-\int_t^T r(X_u)du}  K(\tilde{S}_T)\mid \tilde{S}_t, X_t=i, Y_t,T_{n(t)+1} > T] \\
		& + \int_0^{T-t} E[e^{-\int_t^T r(X_u)du} K(\tilde{S}_T)\mid \tilde{S}_t, X_t, Y_t, T_{n(t)+1}=t + v] \frac{f(t-T_{n(t)}+v\mid X_t)} {1-F(Y_t\mid X_t)}dv\\
		=& \frac{ 1- F(T-T_{n(t)}\mid X_t)}{1-F(Y_t\mid X_t)} \rho_{X_t}(t,\tilde{S}_t)+ \int_0^{T-t} e^{-r(X_t)v} \frac{f(Y_t+v\mid X_t)} {1-F(Y_t\mid X_t)} \times\\ 
		& \sum_{j\neq i} p_{ij}(Y_t +v) \int_0^{\infty} E[e^{-\int_{t+v}^T r(X_u)du} K(\tilde{S}_T)\mid \tilde{S}_{t+v}=x,Y_{t+v}=0,\\
		& X_{t+v}=j, T_{n(t)+1}=t+v] \frac{\exp\{\frac{-1}{2}((\ln(\frac{x}{\tilde{S}_t})-\int_{t}^{t+v}(r(i) -\frac{\si^2(u,i)}{2})\,du) \frac{1}{\sqrt{\int_{t}^{t+v}\si^2(u,i)\,du} })^2\}}{x\sqrt{2\pi}\sqrt{\int_{t}^{t+v}\si^2(u,i)\,du} } dx\, dv\\
		=& \frac{ 1- F(T-t+Y_t\mid X_t)}{1-F(Y_t\mid X_t)} \rho_{X_t}(t,\tilde{S}_t)+\int_0^{T-t} e^{-r(X_t)v} \frac{f(Y_t+v\mid X_t)} {1-F(Y_t\mid X_t)} \times\\
		& \sum_{j\neq i} p_{ij} (Y_t+v) \int_0^{\infty} \varphi(t+v,x,j,0) \frac{e^{\frac{-1}{2}L^2}}{x\sqrt{2\pi}\sqrt{\int_{t}^{t+v}\si^2(u,i)\,du} }\, dx\, dv.
	\end{align*}
	Finally by using irreducibility condition (A1), we can replace $( \tilde{S}_t, X_t,Y_t)$ by generic variable $(s,i,y)$ in the above relation and thus conclude that $\varphi$ is a solution of \eqref{pricing pde}. Thus (i) holds.
	
	(ii) \noi We note that since $ K $ is of at most linear growth, there exist $ k_1,k_2>0 $ such that $ K(s)\leq k_1 + k_2 s $ for all $ s\geq 0 $. Hence,
	\begin{align*}
		\vf(t,\tilde{S}_t,X_t,Y_t)=&\tilde{E}[e^{-\int_t^T r(X_u)\,du}K(\tilde{S}_T)\mid \tilde{\mathcal{F}}_t]\\
		\leq & \tilde{E}[e^{-\int_t^T r(X_u)\,du}(k_1 + k_2 \tilde{S}_T)\mid \tilde{\mathcal{F}}_t]\\
		\leq & k_1 + k_2 \tilde{E}[e^{-\int_t^T r(X_u)\,du}\tilde{S}_T\mid \tilde{\mathcal{F}}_t].
	\end{align*}
	Since $ \vf(t,\tilde{S}_t,X_t,Y_t)=\tilde{E}[e^{-\int_t^T r(X_u)\,du}K(\tilde{S}_T)\mid \tilde{\mathcal{F}}_t] $, using the martingale property of $ e^{-\int_0^t r(X_u)\,du}\tilde{S}_t $, from equation \eqref{eq18} and the above, we get 
	\begin{equation*}
		\vf(t,\tilde{S}_t,X_t,Y_t)\leq k_1+k_2\tilde{S}_t.
	\end{equation*}
	From equation \eqref{eq18}, it is evident that $ \vf $ is an expectation of a non-negative quantity, and hence is non-negative. Thus (ii) holds.
	 \qed
	\begin{theo}\label{theo1}
		The initial-boundary value problem \eqref{p1}-\eqref{boundary} has a unique classical solution in the class of functions with at most linear growth.
	\end{theo}
	\proof Existence follows from Lemma \ref{lm2} and Proposition \ref{theo3}. For uniqueness, first assume that $\vf_1$ and $\vf_2$ are two classical solutions of \eqref{p1}-\eqref{boundary} in the prescribed class. Then using Proposition \ref{theo4}, we know that both also solve \eqref{pricing pde}. But from Lemma \ref{lm2}, there is only one such in the prescribed class. Hence $\vf_1=\vf_2$.\qed
	\begin{rem} The above theorem can also be proved in a different manner which heavily depends on the mild solution techniques \cite{PA} and Proposition 3.1.2 of \cite{CFMW}. Such an alternative approach is taken in \cite{AGMKG} to establish well-posedness of a special case of \eqref{p1}-\eqref{boundary}. The reason for adopting the present approach is that, it enables us to establish the equivalence between the PDE and an IE in the go. This in tern suggests an alternative expression of partial derivative of the solution. In the next section the importance of such representation is explained.
	\end{rem}

	\chapter{The option pricing problem}
	We concern ourselves with an extension of the widely-studied Black-Scholes model of financial markets. In our model, the market exhibits semi-Markov regime-switching. The Markov-modulated regime-switching model has been studied in \cite{tamal}. We use age-dependent processes, which have been discussed in Chapter 2 of this thesis, to extend this model.
	
	\noi Various financial instruments are traded in financial markets. Some of these instruments are stocks, bonds, options, futures, swaps, etc. Financial instruments whose price depend on the price of some other commodity are called derivatives. Options and futures are examples of derivatives.
	
	\noi An option is a contract between two parties- the writer of the option, and the holder of the option. The holder of the option purchases the option from the writer at a premium, called the ``price'' of the option. There are several types of options. The most common are European and American options. These are usually traded on exchanges, and are referred to as ``vanilla'' options. The other kinds of options are not so common, and are called ``exotic'' options. All options are further classified into call options and put options. A European call option confers upon its holder the right to buy a certain amount of stock at a fixed price, called the ``strike price'', at the time of maturity, while a European put option allows its holder to sell the same.
	
	\noi It is obvious that one must pay a premium to purchase an option. Without the premium, the holder of an option would never suffer a loss, violating the no-arbitrage condition which is satisfied in most real-life markets. The premium must be fair to both the holder as well as the writer of the option. The price of an option is thus the expected value of the discounted price of its corresponding contingent claim in a risk-neutral market.

	\noi The Black-Scholes model is a standard model used for pricing European-style options. It makes a number of assumptions, which are stated below:
	\begin{enumerate}
		\item The rate on the riskless asset is constant, and is thus called the risk-free interest rate.
		\item The logarithm of the stock price is a geometric Brownian motion (GBM) with constant drift and volatility.
		\item The stock is dividend-free.
		\item There are no arbitrage opportunities.
		\item It is possible to borrow or lend any amount, even fractional, of cash at the risk-free interest rate.
		\item It is possible to buy or sell any amount, even fractional, of the stock. This includes the possibility of short selling, i.e the act of selling a stock one does not own.
		\item The market is frictionless, i.e devoid of any fees or taxes, etc. 
	\end{enumerate}
	
	\noi The present price of a European call option, in the Black-Scholes model, can be expressed as
	\begin{equation*}
	\eta(t,s)=\tilde{E}[e^{-r(T-t)}(S_T-K)^+\mid S_t=s],
	\end{equation*}
	where $\tilde{E}$ is the risk-neutral measure, $r$ is the risk-free interest rate, $S_t$ is the stock price at the present time $ t $ and $T$ and $K$ are the maturity and the strike price, respectively.
	
	\noi Under the usual notation, the price of a European call option in the Black-Scholes model can also be expressed as the solution to a parabolic partial differential equation, known as the Black-Scholes PDE. This PDE is
	\begin{equation}\label{bs}
		\frac{\partial\eta(t,s)}{\partial t} + rs\frac{\partial\eta(t,s)}{\partial s}+\frac{1}{2}\sigma^2 s^2 \frac{\partial^2\eta(t,s)}{\partial s^2} = r\eta(t,s),
	\end{equation}
	with appropriate terminal conditions.
	
	\noi This PDE is a particular case of \eqref{eq2}, for a fixed $i$, where $r$ and $\sig$ are time-independent.	Equation \eqref{bs} can be solved analytically to give
	\begin{equation}
		\eta(t,s)=N\left( \frac{\ln\left( \frac{s}{K} \right) +\left( r+\frac{\sig^2}{2} \right)(T-t)}{\sig\sqrt{T-t}}\right) s - N\left( \frac{\ln\left( \frac{s}{K} \right) +\left( r-\frac{\sig^2}{2} \right)(T-t)}{\sig\sqrt{T-t}}\right) Ke^{-r(T-t)},
	\end{equation}
	
	
	\noi where  $N(.)$ is the cumulative distribution function of the standard normal distribution.
	
	\noi However, in practice, few of the conditions of the Black-Scholes model are met. Hence, we consider regime-switching models. Section 2.4 has discussed the motivation behind our study of age-dependent processes.
	
	\section{The Market Model}\label{sec 4.1}
	\noi Let $\{B_t\}_{t\geq0}$ be the price of money market account at time $t$ where, spot interest rate is $r_t=r(X_t)$ and $B_0=1$. Here, $ \{X_t\}_{t\geq 0}  $ is taken to be an age-dependent process discussed in Chapter 2. We have $B_t = e^{\int_0^t r(X_{u}) du}$. Let $\{S_t\}_{t\geq 0}$ be the price process of the stock, which is governed by a semi-Markov modulated GBM i.e.,
	\begin{equation}\label{8}
		dS_t = S_t~(\mu(t,X_{t})dt +\sigma(t,X_{t}) dW_t),\tab S_0>0,
	\end{equation}
	where $ \{W_t\}_{t\geq 0}$ is a standard Wiener process independent of $ \{X_t\}_{t\geq 0}$, $\mu : \mathcal{X} \to \mathbb{R}$ is the drift coefficient and $\sigma : [0,T]\times\mathcal{X} \to (0, \infty)$ corresponds to the volatility. Let $\mathcal{F}_t$ be a filtration of $\mathcal{F}$ satisfying usual hypothesis and right continuous version of the filtration generated by $X_t$ and $S_t$. Clearly the solution of the above SDE is an $\mathcal{F}_t$ semimartingale with almost sure continuous paths.
	
	\noi We address the problem of pricing derivatives under the above market assumptions. To this end we recall the quadratic hedging approach in a general market setup below.

	\section{Quadratic Hedging}
	
	Let a market consist of two assets $\{S_t\}_{t\geq0}$ and $\{B_t\}_{t\geq0}$ where $S_t$ and $B_t$ are continuous semi-martingales and $B_t$ is of finite variation.
	An \emph{admissible strategy} is a dynamic allocation to these assets and is defined as a predictable process $\pi=\{\pi_t=(\xi_t,\varepsilon_t), 0\leq t\leq T\}$ which satisfies conditions, given in $(A1)$ below.
	The components $\xi_t$ and $\varepsilon_t$ denote the amounts invested in $S_t$ and $B_t$ respectively at time $t$.
	The value of the portfolio at time t is given by
	\begin{eqnarray}\label{1a}
		V_t = \xi_tS_t + \varepsilon_tB_t.
	\end{eqnarray}
	\noi Here we assume
	\begin{itemize}
		\item[(A1)](i) $\xi_t$ is square integrable w.r.t $S_t$,\\
		(ii) $E(\varepsilon^2_t)<\infty$,\\
		(iii) $\exists a>0$ s.t. $P(V_t\geq-a,t\in[0,T])=1$.
	\end{itemize}
	
	\noi It can be shown, in a similar vein as in \cite{AGMKG}, that the market model under consideration admits the existence of an equivalent martingale measure. Hence, under the class of admissible strategies defined above, the market is free of arbitrage opportunities. This allows us to consider pricing using the F\"{o}llmer-Schweizer decomposition of the contingent claim.
	
	\noi Let $C_t$ be the accumulated additional cash flow due o a strategy $\pi$ at time $t$. Then $V_t$ can
	also be written as sum of two quantities, one is the return of the
	investment at an earlier instant $t-\Delta$ and the other one is the instantaneous cash flow $(\Delta C_t)$.
	\begin{eqnarray}\label{eq.1.1}
		ie. \quad V_t &=& \xi_{t-\Delta}S_t + \varepsilon_{t-\Delta}B_t + \Delta C_t \\
		or \quad \Delta C_t &=& S_t (\xi_t - \xi_{t-\Delta}) + B_t (\varepsilon_t - \varepsilon_{t-\Delta}) \no
	\end{eqnarray}
	which is different from $S_{t-\Delta}(\xi_t - \xi_{t-\Delta}) + B_{t-\Delta}(\varepsilon_t - \varepsilon_{t-\Delta})$. The above observation indicates that the external cash flow can be represented as a stochastic integral (but not in the It\={o} sense) resembling $S_t d\xi_t + B_t d\varepsilon_t$. It would have the same integrator and integrand but would be defined by taking the right end points instead of left end points unlike the It\={o} integral. However, here we confine ourselves in the formalism of It\={o} calculus alone. In order to derive an expression using It\={o} integrals, we note that the equations \eqref{1a} and \eqref{eq.1.1} lead to the following discrete equation
	\begin{eqnarray}
		V_t - V_{t-\Delta} = \xi_{t-\Delta}(S_t - S_{t-\Delta}) + \varepsilon_{t-\Delta}(B_t - B_{t-\Delta}) + \Delta C_t \nonumber
	\end{eqnarray}
	or equivalently the SDE
	\begin{eqnarray}\label{1f}
		dV_t=\xi_t dS_t + \varepsilon_t dB_t + dC_t.
	\end{eqnarray}
	This observation essentially makes the following (see \cite{SAN} for details) definition, which is standard in the literature, self explanatory.
	\noi\begin{defn1}
		A strategy $\pi=(\xi,\varepsilon)$ is defined to be self financing if
		\begin{eqnarray*}
			dV_t = \xi_t dS_t + \varepsilon_t dB_t,\tab \forall t \geq 0.
		\end{eqnarray*}
	\end{defn1}
	\noi Now using integration by parts rule of It\^{o} integration, we deduce from (\ref{1a})
	\begin{eqnarray*}\label{1g}
		dV_t=\xi_t dS_t + \varepsilon_t dB_t + S_t d\xi_t + B_td\varepsilon_t + d\langle S,\xi \rangle_t + d\langle B,\varepsilon\rangle_t.
	\end{eqnarray*}
	\noi By comparing this with equation (\ref{1f}) we get
	\begin{eqnarray}\label{1c}
		dC_t = S_t d\xi_t + B_td\varepsilon_t + d\langle S,\xi \rangle_t + d\langle B,\varepsilon\rangle_t.
	\end{eqnarray}
	\noi Since, $B_t$ is of finite variation and of continuous path, $ \langle B,\varepsilon\rangle_t=0 $ for all $ t $. We further notice that
	
	\begin{align}
		\no d((\xi_t S^*_t)B_t)=& \xi_t S^*_t\,dB_t+B_t\,d(\xi_t S^*_t)+d\langle\xi S^*, B\rangle_t\\
		=& \xi_t S^*_t\,dB_t + B_t S^*_t\,d\xi_t + B_t\xi_t\,dS^*_t + B_t\,d\langle\xi, S^*\rangle_t + d\langle\xi S^*, B\rangle_t \label{chain rule 1}
	\end{align}
	and
	\begin{align}
		\no d(\xi_t(S^*_t B_t))=&\xi_t\,d(S^*_t B_t) + S^*_t B_t\,d\xi_t + d\langle\xi, S^* B\rangle_t\\
		=& \xi_t B_t\,dS^*_t + \xi_t S^*_t\,dB_t + \xi_t\,d\langle S^*,B\rangle_t + S^*_t B_t\,d\xi_t + d\langle\xi, S^* B\rangle_t. \label{chain rule 2}
	\end{align}
	Thus, from equations \eqref{chain rule 1} and \eqref{chain rule 2}, we get
	\begin{equation}
		B_t\,d\langle \xi,S^*\rangle_t + d\langle\xi S^*, B\rangle_t = \xi_t\,d\langle S^*,B\rangle_t + d\langle\xi, S^* B\rangle_t.
	\end{equation}
	
	\noi Thus,
	
	\begin{align*}
		B_td\langle S^*,\xi\rangle_t= & d\langle BS^*,\xi\rangle_t+\xi_td\langle S^*,B\rangle_t-d\langle S^*\xi,B\rangle_t\\
		=& d\langle S,\xi\rangle_t,
	\end{align*}
	where $S^*_t:= B_t^{-1}S_t$. Thus using (\ref{1a}) and above identity, equation (\ref{1c}) gives
	\begin{align*}
		dC_t=& S_t\,d\xi_t + B_t\,d(V_t^*-\xi_t S_t^*)+B_t\,d\langle S^*,\xi\rangle_t\\
		=&{S_td\xi_t}+B_t(dV_t^*-\xi_tdS^*_t-{S^*_td\xi_t}- {d\langle S^*,\xi\rangle_t})+B_t{d\langle S^*,\xi\rangle_t}\\
		=& B_t(dV_t^*-\xi_t dS_t^*)
	\end{align*}
	or,
	\begin{equation}\label{1d}
		\frac{1}{B_t}dC_t = dV^*_t-\xi_tdS^*_t.
	\end{equation}
	\noi The process $C^*_t:=C^*_0 + \int^t_0 \frac{1}{B_t}dC_t$, for obvious reason, is called the \emph{discounted cost process} which gives the net
	present value at $t=0$ of the accumulated additional cash flow up to time $t$. If a strategy $\pi$ is \emph{self-financing}, clearly $C^*_t(\pi)=$ constant and hence one has from \eqref{1d},
	$$ dV^*_t = \xi_tdS^*_t .
	$$
	The Black-Scholes model is an example of what is called a \textit{complete} market. A complete market is one in which all contingent claims are attainable by self-financing strategies. In many market models, the class of self financing strategies is inadequate to ensure a perfect hedge for a given claim. Such markets are called \emph{incomplete}. In such a market an \emph{optimal strategy} is an admissible hedging strategy for which the \emph{quadratic residual risk}, a measure of the cash flow, is minimized subject to a certain constraint(see \cite{FS} for more details). This optimal strategy need not be self-financing. It is shown in \cite{FS} that if the market is arbitrage free, the existence of an optimal strategy for hedging an $\calF_T$ measurable claim $H$, is equivalent to the existence of F\"{o}llmer Schweizer decomposition of discounted claim $H^*:= B^{-1}_T H$ in the form
	\begin{equation}\label{eq1}
		H^*=H_0+\int^{T}_{0}{\xi^{H^*}_t}dS^*_t+L^{H^*}_T,
	\end{equation}
	where $H_0\in L^2(\Omega,\mathcal{F}_0,P), L^{H^*}=\{L^{H^*}_t\}_{0\leq t\leq T}$ is a square integrable martingale starting with zero and orthogonal to the martingale part of $S_t$, and $\xi^{H^*}=\{\xi^{H^*}_t\}_{t\geq0}$ satisfies A1 (i). Further $\xi^{H^*}$ appeared in the decomposition, constitutes the optimal strategy. Indeed the optimal strategy $\pi=(\xi_t,\varepsilon_t)$ is given by
	\begin{eqnarray}\label{11a}
		\nonumber \xi_t &:=& \xi^{H^*}_t,\\
		V^*_t &:=& H_0+\int^{t}_{0}{\xi_u}dS^*_u+L^{H^*}_t,\\
		\nonumber \varepsilon_t &:=& V^*_t-\xi_tS^*_t,
	\end{eqnarray}
	\noi and $B_t V^*_t$ represents the \emph{locally risk minimizing price} at time $t$ of the claim $H$.  The pricing and hedging problems in any market, especially an incomplete one, can thus be addressed by constructing the F\"{o}llmer-Schweizer decomposition of the relevant contingent claim.
	
	\noi Returning to our particular market model as described in Section \ref{sec 4.1}, we aim to construct the F\"{o}llmer-Schweizer decomposition.

	\section{Hedging and Pricing equations}
	 We seek to find an expression for the optimal hedging strategy for a number of European-type options. In this section, we discuss call, put and barrier options. Options can be categorised, depending on their dependence on the path of the stock price process. 
	
	\subsection{Path-independent options}\label{path independent}
	Path-independent options such as European call/put options and their combinations (butterfly spreads, etc.) are the easiest to price.

		\begin{theo}\label{theo5} Let $\vf$ be the unique classical solution of \eqref{p1}-\eqref{boundary} in the class of functions with at most linear growth.
		\begin{enumerate}
			\item Let $(\xi,\varepsilon)$ be given by
			\begin{equation}\label{VI3.20}
				\xi_t :=\frac{\partial\varphi(t,S_t,X_{t-},Y_{t-})}{\partial s} \txt{ and } \varepsilon_{t} := e^{-\int_{0}^{t}r(X_{u})du} (\varphi(t,S_t,X_{t},Y_{t})-\xi_{t}S_{t}).
			\end{equation}
			Then $(\xi,\varepsilon)$ is the optimal admissible strategy.
			\item $\varphi(t,S_t,X_t,Y_t)$ is the locally risk minimizing price of $K(\tilde{S}_T)$.
		\end{enumerate}
		\end{theo}
						
		\proof Under the market model, the mean variance tradeoff (MVT) process $\hat{K}_t$ (as defined in Pham et al \cite{PH}) takes the following form
		\begin{equation*}
			\hat{K}_t=\int_0^t\left(\frac{\mu(s,X_s)-r(X_s)}{\sigma(s,X_s)}\right)^2 ds.
		\end{equation*}
		Hence $\hat{K}_t$ is bounded and continuous on $[0,T]$. We also know that $S_t$ has almost sure continuous paths. Since, $H^*\in L^2(\Omega, {\cal F}, P)$ for $H=K(\tilde{S}_T)$ we apply corollary 5 and Lemma 6 of \cite{PH} to conclude that $H^*$ admits a F\"{o}llmer-Schweizer decomposition
		\begin{equation}\label{follschweiz}
			H^*=H_0+\int_0^T \xi_u^{H^*}\left( dA_u^*+\kappa(u,X_u)A^*_u\,du \right) +L^{H^*}_T,
		\end{equation}
	
		  \noi with an integrand $\xi^{H^*}$ satisfying A1 (i) and $L^{H^*}$ being square integrable. Therefore, to prove the theorem it is sufficient to show that
		\begin{itemize}
			\item[(a)] there exists $\mathcal{F}_0$ measurable $H_0$ and $\mathcal{F}_T$ measurable $L_T$ such that $L_t:= E[L_T\mid \mathcal{F}_t]$ is orthogonal to $\int_0^t \sig(X_t)S^*_t dW_t$ i.e., the martingale part of $S^*_t$ and $H^*=H_0+\int^{T}_{0}{\xi_t}dS^*_t+L_T$;
			\item[(b)] $\frac{1}{B_t} \varphi(t,S_t,X_{t-},Y_{t-})= H_0+\int^{t}_{0}{\xi_t}dS^*_t+L_t$ for all $t\le T$;
			\item[(c)] $\varphi(t,S_t,X_{t},Y_{t})= B_t \varepsilon_{t} + \xi_{t}S_{t}$ for all $t\le T$;
			\item[(d)] $P(\varphi(t,S_t,X_{t},Y_{t}) \ge 0 \forall t\le T)=1$,
		\end{itemize}
		where $\vf$ is the unique classical solution of \eqref{p1}-\eqref{boundary} in the prescribed class and $(\xi,\varepsilon)$ is as in \eqref{VI3.20}.
							
		\noi In Lemma \ref{lm2} it is shown that $\vf$ is a non-negative function. Hence (d) holds. From the definition of $\varepsilon_t$ in \eqref{VI3.20}, (c) follows. Next we show the condition (b). We apply It\^{o}'s formula to $e^{-\int_{0}^{t}r(X_{u})du} \varphi(t,S_{t},X_{t},Y_{t})$ under the
		measure $P$ to get
		\begin{align}
			\no e^{-\int_{0}^{t}r(X_{u})\,du}\varphi(t,S_{t},X_{t},Y_{t})=& \varphi(0,S_{0},X_{0},Y_{0})+\int_0^t e^{-\int_0^u r(X_v)\,dv}\frac{\pa\vf}{\pa u}(u,S_u,X_{u-},Y_{u-})\,du\\
			\no&+ \int_{0}^T e^{-\int_0^u r(X_v)\,dv}\left(-r(X_u)\right)\vf(u,S_u,X_{u-},Y_{u-})\,du\\
			\no&+ \int_{0}^T e^{-\int_0^u r(X_v)\,dv}\frac{\pa\vf}{\pa s}(u,S_u,X_{u-},Y_{u-})\,dS_u\\
			\no&+ \frac12 \int_{0}^T e^{-\int_0^u r(X_v)\,dv}\frac{\pa^2\vf}{\pa s^2}(u,S_u,X_{u-},Y_{u-})\,d\langle S\rangle_u\\
			\no&+ \int_{0}^T e^{-\int_0^u r(X_v)\,dv}\frac{\pa\vf}{\pa y}(u,S_u,X_{u-},Y_{u-})\,d Y^{(c)}_u\\
			&+ \sum_{u\leq t}e^{-\int_0^u r(X_v)\,dv}\left( \vf(u,S_u,X_{u},Y_{u})-\vf(u,S_u,X_{u-},Y_{u-}) \right),\label{ito on vf}
		\end{align}
		where $ Y^{(c)}_t $ is the continuous part of $ Y_t $. Now,
		\begin{align*}
			\vf(u,S_u,X_{u},Y_{u})-\vf(u,S_u,X_{u-},Y_{u-})=& \vf\left(u,S_u,X_{u-}+\int_{\mathbb{R}}h(X_{u-},Y_{u-},z)\,\wp(du,dz),\right.\\
			 &\left. Y_{u-} - \int_{\mathbb{R}} g(X_{u-},Y_{u-},z)\,\wp(du,dz)\right)\\
			 &- \vf(u,S_u,X_{u-},Y_{u-})\\
			 =& \int_{\mathbb{R}} \left[ \vf\left(u,S_u,X_{u-}+ h(X_{u-},Y_{u-},z), Y_{u-} - g(X_{u-},Y_{u-},z) \right)\right.\\
			 &-\left. \vf(u,S_u,X_{u-},Y_{u-})\right]\,\wp(du,dz)\\
			 =& \int_{\mathbb{R}} \left[ \vf\left(u,S_u,X_{u-}+ h(X_{u-},Y_{u-},z), Y_{u-} - g(X_{u-},Y_{u-},z) \right)\right.\\
			 &-\left. \vf(u,S_u,X_{u-},Y_{u-})\right]\,(\hat{\wp}(du,dz)+\,du\,dz),\\
		\end{align*}
		where $ \hat{\wp} $ is the compensated Poisson random measure.  We set
		\begin{eqnarray*}
		L_t&:=& \int_{0}^{t}e^{-\int_{0}^{u}r(X_{v})dv}\int_{\mathbb{R}} [\varphi(u,S_{u},X_{u-} +h(X_{u-},Y_{u-},z),Y_{u-}-g(X_{u-},Y_{u-},z)) \\
		&&-\varphi(u,S_{u},X_{u-},Y_{u-})]{\hat{\wp}}(du,dz).
		\end{eqnarray*}
		From the definitions of $h$ and $g$, we can write
		\begin{align*}
			X_{u-}+ h(X_{u-},Y_{u-},z)=& \sum_{j\neq X_{u-}} j \mathds{1}_{\Lambda_{X_{u-}j}(Y_{u-})}(z) + X_{u-}\mathds{1}_{\bigcup_{j\neq i}\Lambda_{X_{u-}j}(Y_{u-})^c }(z)
		\end{align*}
		and
		\begin{align*}
			Y_{u-}- g(X_{u-},Y_{u-},z)=&  Y_{u-}\mathds{1}_{\bigcup_{j\neq i}\Lambda_{X_{u-}j}(Y_{u-})^c }(z).
		\end{align*}
		Thus, 
		\begin{align*}
			&\int_{\mathbb{R}} \left[ \vf\left(u,S_u,X_{u-}+ h(X_{u-},Y_{u-},z), Y_{u-} - g(X_{u-},Y_{u-},z) \right)-\vf(u,S_u,X_{u-},Y_{u-})\right]\,du\,dz\\
			=& \sum_{X_{u-}\neq j} \left[ \vf(S_u,j,0)-\vf(S_u,X_{u-},Y_{u-}) \right] \lambda_{X_{u-}j}(Y_{u-})\,du.
		\end{align*}
		
		\noi We know that $ dS_t=S_t(\mu(t,X_t)\,dt + \sigma(t,X_t)\,dW_t $, $ d\langle S\rangle_t = \sigma^2(t,X_t)\,dt $ and $ dY^{(c)}_t=dt $. Hence, from \eqref{ito on vf}, we get
		\begin{align}
			\no e^{-\int_{0}^{t}r(X_{u})\,du}\varphi(t,S_{t},X_{t},Y_{t})=& \varphi(0,S_{0},X_{0},Y_{0})+\int_0^t e^{-\int_0^u r(X_v)\,dv}\times\\
			\no& \left( \frac{\pa \vf}{\pa u}+\frac{\pa \vf}{\pa y} + \mu(u,X_u)S_u \frac{\pa \vf}{\pa s} +\frac12 \sigma^2(u,X_u) S_u^2 \frac{\pa^2 \vf}{\pa s^2}-r(X_u)\right.\\
			\no&+\left. \sum_{X_{u-}\neq j} \left[ \vf(S_u,j,0)-\vf(S_u,X_{u-},Y_{u-}) \right] \lambda_{X_{u-}j}(Y_{u-}) \right)\,du\\
			\no&+ \int_0^t e^{-\int_0^u r(X_v)\,dv} \sigma(u,X_u)\frac{\pa\vf}{\pa s}(u,S_u,X_{u-},Y_{u-})\,dW_u + L_t.
		\end{align}
		
		\noi Using \eqref{p1}, this simplifies to
		\begin{align*}
			& \varphi(0,S_{0},X_{0},Y_{0})+\int_0^t e^{-\int_0^u r(X_v)\,dv} (\mu(u,X_u)-r(X_u))S_u\frac{\pa\vf}{\pa s}\,du \\
			&+ \int_0^t e^{-\int_0^u r(X_v)\,dv} \sigma(u,X_u)\frac{\pa\vf}{\pa s}(u,S_u,X_{u-},Y_{u-})\,dW_u + L_t.
		\end{align*}

		\noi Now, $ S_t^*=B_t^{-1}S_t=e^{-\int_0^t r(X_u)\,du}S_t $. Hence 
		 \begin{equation*}
		 dS_t^*=e^{-\int_0^t r(X_u)\,du}S_t((\mu(t,X_t)-r(X_t))\,dt + \sigma(t,X_t)\,dW_t.
		 \end{equation*}
		
		\noi Thus, we obtain, for all $t<T$
		\begin{align*}\label{VI3.24}
			e^{-\int_{0}^{t}r(X_{u})\,du} \varphi(t,S_{t},X_{t},Y_{t}) = &
			\varphi(0,S_{0},X_{0},Y_{0})+\int_{0}^{t}\frac{\partial
			\varphi(u,S_{u},X_{u-},Y_{u-})}{\partial s} d S^*_{u}\no
			+ L_t.
		\end{align*}
	
		\noi Since, $L_t$ is an integral w.r.t. a compensated Poisson random measure, it is a martingale. Again the independence of $W_t$ and $\wp$ implies the orthogonality of $L_t$ to the martingale part of $S^*_t$. Thus, we obtain the following F-S decomposition by letting $t \uparrow T $,
		\begin{equation}\label{VI3.25}
			B_T^{-1}K(\tilde{S}_T) = \varphi(0,S_{0},X_{0},Y_{0}) + \int^{T}_{0}{\xi_t}dS^*_t + L_T.
		\end{equation}
		Thus (a) and (b) hold.\qed

		\begin{theo}\label{theo6}
			Let $\varphi$ be the unique solution of \eqref{p1}-\eqref{boundary}. Set
			\begin{eqnarray}\label{5}
			\no\psi(t,s,i,y) &:=& \frac{ 1- F(T-t+y\mid i)}{1-F(y\mid i)} \frac{\partial\eta_{i}(t,s)}{\partial s}+\int_0^{T-t}e^{-r(i)v} \frac{f(y+v\mid i)} {1-F(y\mid i)} \times\sum_j p_{i j}(y+v) \\
			\no&& \int_0^{\infty} \varphi(t+v,x,j,0) \frac{e^{\frac{-1}{2}L(t,i)^2}}{\sqrt{2\pi} xs\bar{\sigma}}\frac{\left(\ln(\frac{x}{s})-(r(i)v -\bar{\sigma}^2)\right)}{\bar{\sigma}^2} dx dv\\
			\end{eqnarray}
			where $(t,s,i,y)\in \mathcal{D}$ and $\bar{\sigma}^2=\int_t^{t+v}\sigma(u,i)^2\,du$.
			Then $\psi(t,s,i,y) =\frac{\pa}{\pa s}\vf(t,s,i,y)$.
		\end{theo}
		
		\proof We need to show that $\psi$ (as in \eqref{5}) is equal to $\frac{\partial\varphi}{\partial s}$. Indeed, one obtains the RHS of \eqref{5} by differentiating the right side of \eqref{p1} with respect to $s$. Hence the proof. \qed
		
		\begin{remark}\label{rem1}
			
			\noi We have shown that $ \frac{\pa}{\pa s}\vf(t,s,i,y) $ is a necessary quantity to be calculated in order to find the optimal hedging. Attempting to compute $ \frac{\pa}{\pa s}\vf(t,s,i,y) $ using numerical differentiation would increase the sensitivity of $ \frac{\pa \vf}{\pa s} $ to small errors. Equation \eqref{5} gives a better, more robust approach to computing $ \frac{\pa}{\pa s}\vf(t,s,i,y) $, using numerical integration.
		\end{remark}

%

	\subsection{Weakly path-dependent options}\label{barrier options}
	In this subsection, we consider barrier options. These are the options which are either exercised or allowed to expire immediately upon the stock price hitting a certain ``barrier''. There are four types of European barrier options (the barrier is assumed to be $b>0$):
	\begin{enumerate}
		\item Down-and-out: The option becomes worthless if the barrier $ S=b $ is reached from above before expiry. 
		\item Up-and-out: The option becomes worthless if the barrier $ S=b $ is reached from below before expiry.
		\item Down-and-in: The option becomes worthless unless the barrier $ S=b $ is reached from above before expiry.
		\item Up-and-in: The option becomes worthless unless the barrier $ S=b $ is reached from below before expiry.
	\end{enumerate}
	The payoff function for barrier options is not solely determined by the stock price at maturity. The option expires, or is immediately exercised (as the case may be), depending on whether the stock price process, $S_t$, hits a certain barrier or not. In other words, the payoff is path-dependent. However, the payoff does not depend on the entire history of the stock price; it only depends on a particular attribute of the stock price process. Thus, barrier options are called ``weakly path-dependent''.
	
	\noi These barrier conditions can apply to call options as well as put options. We consider the problem of pricing an up-and-out European call option in this subsection. We, however, restrict ourselves to the case where the volatility does not depend explicitly on time, so that $\sigma(t,i)=\sigma(i)$ for all $ i$ and $t$. Let the price of the up-and-out European call option be $ \vf_c^{uo} $. Then, the contingent claim can be written as
	\begin{equation}
		H=(S_T-K)^+\mathds{1}\left \lbrace\max_{t\in [0,T]}S_t <b\right\rbrace,
	\end{equation} 
	under the usual notation. We define $ \tau:=\min\{t>0:S_t=b\} $. Thus, $\tau$ is an $ \mathcal{F}_t$-stopping time, which is almost surely finite. Now, if $ S_0\geq b $, then the option will already be in a state of expiry. Hence, we only consider the non-trivial case $S_0<b$. In this case, the contingent claim $H$ can be written in an alternative form as
	\begin{equation*}
		H=(S_T-K)^+\mathds{1}\left \lbrace \tau>T \right\rbrace.
	\end{equation*} 
	The pricing problem for barrier options reduces	to the one of solving equation \eqref{pricing pde} on the domain
	\begin{equation*}
		\mathcal{D}_-:=\{(t,s,i,y)\in(0,T)\times(0,b)\times\chi\times(0,T)\}
	\end{equation*}
	with the boundary condition
	\begin{equation}\label{barrier boundary}
		\vf(t,b,i,y)=0\ \text{for all } t\in(0,T),\ i\in\chi.
	\end{equation}
	The analysis we have made in Section \ref{uniqueness} regarding the redundancy of the boundary condition as $s\downarrow 0$ does not apply here, for $ s\uparrow b $, because the pricing PDE does not reduce to an $s$-independent PDE. Hence, the boundary condition \eqref{barrier boundary} is necessary.
	\begin{lem}
		Consider the following integral equation
		\begin{align}
		\no\varphi_c^{uo}(t,s,i,y)=& \frac{ 1- F(T-t+y\mid i)}{1-F(y\mid i)} \eta^{uo}_{c;i}(t,s)+\int_0^{T-t} e^{-r(i)v} \frac{f(y+v\mid i)} {1-F(y\mid i)} \times\\
		&\no\left[ \Phi\left( \dfrac{\ln\left (\frac{b}{s}\right )-(r(i) -\frac{\si^2(i)}{2})v}{\sigma(i)\sqrt{v}} \right)-\exp\left\lbrace \left( \frac{2r(i)}{\sigma^2(i)}-1 \right)\ln\left (\frac{b}{s}\right ) \right\rbrace\times \right.\\
		&\left. \Phi\left( \dfrac{-\ln\left (\frac{b}{s}\right )-(r(i) -\frac{\si^2(i)}{2})v}{\sigma(i)\sqrt{v}} \right)  \right]\times\no\\
		&\no \sum_{j\neq i} p_{ij} (y+v) \int_0^b \varphi_c^{uo}(t+v,x,j,0) \alpha(x;s,i,v)\, dx\, dv.\label{barrier integral}\\
		\end{align}
		Then (i) equation \eqref{p1} has a unique solution $ C(\overline{\mathcal{D_-}}) $, (ii) the solution of the integral equation is in $ C^{1,2,1}(\mathcal{D})$, and (iii) $\vf(t,s,i,y)$ is non-negative.
	\end{lem}
	\begin{proof}
		The proof is similar to that of Lemma \ref{lm2}.\qed
	\end{proof}

	\begin{prop}
		The unique solution of equation \eqref{barrier integral} also solves the initial value problem \eqref{p1}-\eqref{barrier boundary}.	
	\end{prop} 
	\begin{proof}
		The proof is similar to that of Proposition \ref{theo3}, albeit slightly less tedious, since $ \vf_c^{uo} $ is a bounded function, and also because $ \sigma(t,i)=\sigma(i) $ for all $ t $ and $ i $.\qed
	\end{proof}

	\begin{prop}
		Assume \eqref{lambda condition 1} and \eqref{lambda condition 2}. We also assume that the transition matrix $ \tilde{p}_{ij}:=\int_0^\infty p_{ij}(y)\,dF_i(y) $ is irreducible. Let $\varphi$ be a classical solution of \eqref{p1}-\eqref{barrier boundary}. Then  $\varphi$ solves the integral equation \eqref{barrier integral}. 
	\end{prop}
	\begin{proof}
		Much of the proof is similar to that of Proposition \ref{theo4}. We construct $ \tilde{S}_t $ as given there. Now if $\vf^{uo}_c$ is the classical solution of \eqref{p1}-\eqref{barrier boundary} then by using the It\^{o}'s formula on $N_t := e^{-\int_0^t r(X_u)du} \varphi^{uo}_c(t,\tilde{S}_t,X_t,Y_t)$, we get
		\begin{eqnarray*}
		dN_t &=& e^{-\int_0^t r(X_u)du}\left(- r(X_t) \varphi^{uo}_c (t,\tilde{S}_t,X_t,Y_t)+ \frac{\partial \varphi^{uo}_c}{\partial t }(t,\tilde{S}_t,X_t,Y_t) + \mathcal{A}_t \varphi^{uo}_c (t,\tilde{S}_t,X_t,Y_t)\right)dt + dM_t
		\end{eqnarray*}
		where $M_t$ is a local martingale. 
		
		\noi Since $ \tilde{S}_t $ is a martingale and $ \vf^{uo}_c $ is a bounded function, $\{N_t\}_t$ is a martingale. Hence
		\begin{align*} 
		\varphi^{uo}_c(t,\tilde{S}_t,X_t,Y_t)=& e^{\int_0^t r(X_u)du}N_t\\
		 =& E[e^{\int_0^t r(X_u)du}N_T\mid \mathcal{ F}_t]\\
		 =& E[e^{-\int_t^T r(X_u)du} K(\tilde{S}_T)\mathds{1}\{\tau>T\}\mid \tilde{S}_t,X_t,Y_t].
		\end{align*}
		By conditioning at transition times and using the conditional lognormal distribution of $\tilde{S}_t$, we get
		\begin{align*}
			\lefteqn{\varphi_c^{uo}(t,\tilde{S}_t,X_t,Y_t)}\\
			=& E[E[e^{-\int_t^T r(X_u)du} K(\tilde{S}_T)\mathds{1}\{\tau>T\} \mid \tilde{S}_t, X_t=i, Y_t, T_{n(t)+1}]\mid \tilde{S}_t, X_t=i, Y_t]\\
			=& P(T_{n(t)+1} > T\mid X_t,Y_t) E[e^{-\int_t^T r(X_u)du}  K(\tilde{S}_T)\mathds{1}\{\tau>T\}\mid \tilde{S}_t, X_t=i, Y_t,T_{n(t)+1} > T] \\
			& + \int_0^{T-t} E[e^{-\int_t^T r(X_u)du} K(\tilde{S}_T)\mathds{1}\{\tau>T\}\mid \tilde{S}_t, X_t, Y_t, T_{n(t)+1}=t + v] \frac{f(t-T_{n(t)}+v\mid X_t)} {1-F(Y_t\mid X_t)}dv\\
			=& \frac{ 1- F(T-T_{n(t)}\mid X_t)}{1-F(Y_t\mid X_t)} \eta^{uo}_{c;X_t}(t,\tilde{S}_t)+ \int_0^{T-t} e^{-r(X_t)v} \frac{f(Y_t+v\mid X_t)} {1-F(Y_t\mid X_t)} \times\\ 
			& \sum_{j\neq i} p_{ij}(Y_t +v) \int_0^{\infty} E[e^{-\int_{t+v}^T r(X_u)du} K(\tilde{S}_T)\mathds{1}\{\tau>T\}\mid \tilde{S}_{t+v}=x,Y_{t+v}=0,\\
			& X_{t+v}=j, T_{n(t)+1}=t+v] \frac{\exp\{-\frac{1}{2}\left ((\ln(\frac{x}{\tilde{S}_t})-(r(i) -\frac{\si^2(i)}{2})v ) \frac{1}{\sigma(i)\sqrt{v} }\right )^2\}}{x\sqrt{2\pi}\si(i)\sqrt{v}}  dx\, dv,
		\end{align*}
		where $ \eta^{uo}_{c;X_t}(t,\tilde{S}_t) $ is the Black-Scholes price of a European  up-and-out call option with constant interest rate $r(i)$ and time-independent volatility $\sigma(i)$. Thus,
		\begin{align*}
			\lefteqn{\varphi_c^{uo}(t,\tilde{S}_t,X_t,Y_t)}\\
			=& \frac{ 1- F(T-t+Y_t\mid X_t)}{1-F(Y_t\mid X_t)} \eta^{uo}_{c;X_t}(t,\tilde{S}_t)+\int_0^{T-t} e^{-r(X_t)v} \frac{f(Y_t+v\mid X_t)} {1-F(Y_t\mid X_t)} \times\\
			& E[e^{-\int_{t+v}^T r(X_u)du} K(\tilde{S}_T)\mid \tilde{S}_{t+v}=x,Y_{t+v}=0, X_{t+v}=j, T_{n(t)+1}=t+v, \tau>T]\times\\
			& P[\tau>T\mid \tilde{S}_{t+v}=x,Y_{t+v}=0, X_{t+v}=j, T_{n(t)+1}=t+v]\times\\
			& \sum_{j\neq i} p_{ij} (Y_t+v) \int_0^b \varphi_c^{uo}(t+v,x,j,0) \frac{e^{\frac{-1}{2}L^2}}{x\sqrt{2\pi}\sigma(i)\sqrt{v}  } dx\, dv.\\
		\end{align*}

		\noi It can be proved, using the reflection principle, that 
		\begin{align*}
			& P\left[\max_{[t,T]}S_u<b \mid \tilde{S}_{t+v}=x,Y_{t+v}=0, X_{t+v}=j, T_{n(t)+1}=t+v\right ]\\
			=& \left[ \Phi\left( \dfrac{\ln\left (\frac{b}{s}\right )-(r(i) -\frac{\si^2(i)}{2})v}{\sigma(i)\sqrt{v}} \right)-\exp\left\lbrace \left( \frac{2r(i)}{\sigma^2(i)}-1 \right)\ln\left (\frac{b}{s}\right ) \right\rbrace\times \right.\\
			&\left. \Phi\left( \dfrac{-\ln\left (\frac{b}{s}\right )-(r(i) -\frac{\si^2(i)}{2})v}{\sigma(i)\sqrt{v}} \right)  \right], 
		\end{align*}
		which means
		
		\begin{align}
			 \no\varphi_c^{uo}(t,\tilde{S}_t,X_t,Y_t)=& \frac{ 1- F(T-t+Y_t\mid X_t)}{1-F(Y_t\mid X_t)} \eta^{uo}_{c;X_t}(t,\tilde{S}_t)+\int_0^{T-t} e^{-r(X_t)v} \frac{f(Y_t+v\mid X_t)} {1-F(Y_t\mid X_t)} \times\\
			 &\no\left[ \Phi\left( \dfrac{\ln\left (\frac{b}{s}\right )-(r(i) -\frac{\si^2(i)}{2})v}{\sigma(i)\sqrt{v}} \right)-\exp\left\lbrace \left( \frac{2r(i)}{\sigma^2(i)}-1 \right)\ln\left (\frac{b}{s}\right ) \right\rbrace\times \right.\\
			 &\left. \Phi\left( \dfrac{-\ln\left (\frac{b}{s}\right )-(r(i) -\frac{\si^2(i)}{2})v}{\sigma(i)\sqrt{v}} \right)  \right]\times\no\\
			 &\no \sum_{j\neq i} p_{ij} (Y_t+v) \int_0^b \varphi_c^{uo}(t+v,x,j,0) \alpha(x;s,i,v)\, dx\, dv.\\
		\end{align}
		Due to the irreducibility condition (A1), we can replace $ \tilde{S}_t $, $ X_t $, and $ Y_t $ by $ s $, $ i $, and $y$, respectively. \qed
		
	\end{proof}
	
	\begin{theo}
		The initial-boundary value problem \eqref{p1}-\eqref{barrier boundary} has a unique classical solution in the class of functions with at most linear growth.
	\end{theo}
	\begin{proof}
		The proof is similar to that of Theorem \ref{theo1}.\qed
	\end{proof}

	\begin{theo}\label{theo7}
		Let $ \vf_c^{uo}(t,s,i,y) $ denote the unique solution of the problem (\ref{pricing pde},\ref{barrier boundary}). Then the following statements hold true:
		\begin{enumerate}
			\item $ \vf_c^{uo}(t,s,i,y) $ is the locally risk-minimizing option price at time $t$ for an up-and-out European call option with strike price $K$, barrier $b>K$ and maturity $T>t$.
			\item An optimal hedging strategy $ \pi^*=\{\xi^*_t,\eta^*_t\} $ is given  by
			\begin{align}
				\no\xi^*_t=&\frac{\partial}{\partial s}\vf_c^{uo}(t,S_t,X_{t-},Y_{t-})\mathds{1}(\tau>T)\\
				\eta^*_t=&V_t^*-\xi^*_t S^*_t,
			\end{align}
			where
			\begin{align*}
				V_t^*=& \vf_c^{uo}(0,S_0,X_0,Y_0)+\int_{0}^{t}\frac{\partial}{\partial s}\vf_c^{uo}(u,S_u,X_{u-},Y_{u-})\mathds{1}(\tau>u)\,dS^*_u\\
				+& \int_0^t \int_{\mathbb{R}}e^{-\int_0^u r(X_v)\,dv}\left \lbrace \vf^{uo}_c\left( u,S_u,X_{u-}+h(X_{u-},Y_{u-},z),Y_{u-}-g(X_{u-},Y_{u-},z) \right) \right.\\
				&-\left. \vf_c^{uo}(u,S_u,X_{u-},Y_{u-}) \right \rbrace \mathds{1}(\tau>u) \, \hat{\wp}(du,dz).
			\end{align*}
			\item The residual risk at time $t$ is given by
			\begin{align}
				\no R_t(\pi^*)=& E\left[ \int_t^T e^{-2\int_0^u r(X_v)\,dv} \frac{f(Y_u|X_u)}{1-F(Y_u|X_u)}\times\right.\\
				&\left. \sum_{j\neq X_u}p_{X_u,j}\left( \vf^{uo}_c(u,S_u,j,0)-\vf^{uo}_c(u,S_u,X_u,Y_u) \right)^2 \mathds{1}(\tau>u)\,du \Bigg | \mathcal{F}_t \right].
			\end{align}
		\end{enumerate}
	\end{theo}
	\begin{proof}
		Let $0\leq t\leq T$. We define
		\begin{align*}
			N_t:=& e^{-\int_0^t r(X_u)\,du}\vf_c^{uo}(t,S_t,X_{t-},Y_{t-})\mathds{1}(\tau>T)\\
			=& e^{-\int_0^{t\wedge\tau} r(X_u)\,du}\vf_c^{uo}({t\wedge\tau},S_{t\wedge\tau},X_{{t\wedge\tau}},Y_{{t\wedge\tau}}),
		\end{align*}
		since $ \vf_c^{uo}(\tau,S_\tau,X_\tau,Y_\tau)=0 $. By It\={o}'s formula, we obtain, under $P$,
		\begin{align}
			\no N_t=& \vf_c^{uo}(0,S_0,X_0,Y_0)+\int_{0}^{t}\frac{\partial}{\partial s}\vf_c^{uo}(u,S_u,X_{u-},Y_{u-})\mathds{1}(\tau>u)\,dS^*_u\\
			\no+&\int_{0}^{t\wedge\tau}\int_{\mathbb{R}}e^{-\int_0^u r(X_v)\,dv}\left \lbrace \vf^{uo}_c\left( u,S_u,X_{u-}+h(X_{u-},Y_{u-},z),Y_{u-}-g(X_{u-},Y_{u-},z) \right) \right.\\
			&-\left. \vf_c^{uo}(u,S_u,X_{u-},Y_{u-}) \right \rbrace  \, \hat{\wp}(du,dz).\label{fs barrier}
		\end{align}
		By Doob's option sampling theorem, the R.H.S of \eqref{fs barrier} is an $ \mathcal{F}_t $-martingale under $P$, which is orthogonal to $ \{M_t\} $ (owing to the independence of $\{W_t\}$ and $ \hat{\wp}(\cdot,\cdot) $). Thus, as $t\uparrow T$, equation \eqref{fs barrier} provides the F\"{o}llmer-Schweizer decomposition of $N_T$ (i.e, the discounted contingent claim). Hence, the propositions in Theorem \ref{theo7} follow immediately. \qed
	\end{proof}
	
	\section{An example of a volatility model}
	There are many different ways in which the volatility can be modelled. Based on empirical data, several models of volatility can be constructed. We consider, in this section, a kind of ``Monday effect'', which is a surge in the volatility of stocks on Monday, due to the two non-trading days preceding it. The volatility can also be assumed to drop throughout the course of a typical week, only to increase sharply at the beginning of the trading week. One of the models which captures this effect is the following:
	\begin{equation*}
	\sigma(t,i)=\sigma(0,i)\left [\alpha+4(1-\alpha)\!\left( t^\beta-\frac12
	\right)^2 \right ],
	\end{equation*}
	where $t$ is the time in weeks and $\alpha$ and $\beta$ are parameters with $ 0<\alpha<1 $ and $ \beta>0 $.	This model assumes the volatility to decrease to a level $ \alpha $ times its maximum value, before  jumping back up. The minimum volatility is attained at $t=(\frac12)^{\frac{1}{\beta}}$. In this model, higher values of $ \alpha $ indicate lower variation in the volatility, while $ \beta $ dictates the position of the volatility trough, with higher values of $ \beta $ leading to later troughs.
	
	\noi Here is an example of the volatility model with $ \sig(0,1)=0.2,\ \sig(0,2)=0.5 $ and $ \sig(0,3)=0.3 $, with parameters $\alpha=\frac12$ and $ \beta=3 $.
	\begin{figure}[h]
		\centering
		\includegraphics[width=0.5\linewidth]{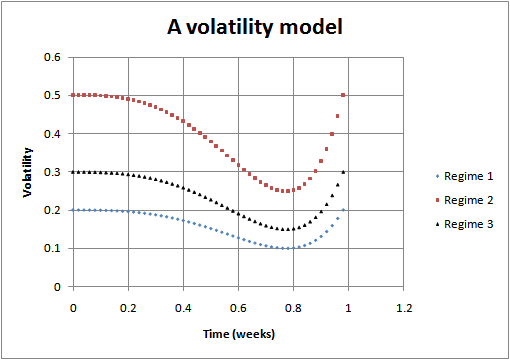}
		\caption{Volatility vs. time}
		\label{fig:volatility}
	\end{figure}

	
	\chapter{Defaultable bonds}
	
	\section{The Market Model}
	We consider a market on a probability space $(\Omega,\mathcal{F},P)$, with a finite state space $ \chi=\{1,2,\dots,k\} $. The market dynamics are modelled by an age-dependent process $X=\{X_t\}_{t\geq 0}$ on $\chi$, as described by equations \eqref{xdef} and \eqref{ydef}.
	We define the following market parameters as the functions
	\begin{equation}\label{market parameters}
		r:\chi\rightarrow(0,\infty),\ \mu:(0,\infty)\times\chi\rightarrow(0,\infty),\ \kappa:(0,\infty)\times\chi\rightarrow\mathbb{R},\ \sigma:(0,\infty)\times\chi\rightarrow(0,\infty).
	\end{equation}
	Here, $r, \mu, \kappa, \sigma$ are the interest rate, the drift coefficient, the dividend payout rate and the volatility, respectively.
	
	\noi We consider a structural model of the company's bond, in which the company defaults on its bond if its asset value drops below a certain threshold. The company's asset value, $A_t$, is assumed to follow a geometric Brownian motion modulated by an age-dependent process $X_t$ given by equations \eqref{xdef} and \eqref{ydef}. Thus,
	\begin{equation}\label{asset}
		dA_t=A_t\left[ \left( \mu(t,X_t)-\kappa(t,X_t) \right)\,dt+\sigma(t,X_t)\,dW_t  \right] ,\ A_0>0
	\end{equation}
	where $\{W_t\}_{t\geq 0}$ is a standard Wiener process independent of $X$.
	The market is also assumed to contain an amount $B_t$ a locally risk-free money-market account, where
	\begin{equation}
		B_t=e^{\int_{0}^{t}r(X_u)\,du}.
	\end{equation}

	\noi We use the structural approach to model the credit risk, i.e the risk of the company defaulting on its debt (bonds). We regard the firm's equity as well as the defaultable bond as contingent claims on the firm's assets. The equity and the debt of the company are denoted by $E_t$ and $D_t$, respectively.

	\subsection{Model 1}
	The first model that we consider is Merton's classical model (\cite{merton}), with a few modifications to account for the fact that the market is modulated by an age-dependent process. We consider a coupon-free bond that can default only on maturity ($t=T$). In the event of  a default, the creditors are entitled to the firm's assets under consideration. Hence, the firm's equity holders receive a payoff only if $A_T>K$, where $K$ is a certain threshold. The total payoff, at maturity, to the equity holders, is
	\begin{equation}
		E(T,A_T,X_T)=(A_T-K)^+=\max(A_T-K,0).
	\end{equation}
	The price of the defaultable bond at maturity is given by
	\begin{equation}
		D(T,A_T,X_T)=\min(A_T,K)=K-(K-A_T)^+.
	\end{equation}
	Since the above payoff is the same as that of a portfolio consisting of a default-free loan with face value $K$, maturing at time $T$ and a short European put option on $A_t$ with dividend rate $\kappa(t,X_t)$, strike price $K$ and maturing at time $T$, it suffices to solve the problem of pricing European call options under the same market model. We have done that in \ref{path independent}. Therefore, we do not produce any further details here.
	
	\subsection{Model 2}
	Merton's classical model does not allow a premature default. It may be that there is a critical threshold below which the firm would be disposed to default on its debt. Such a model is more favourable to the owners of the defaultable bonds. We consider a model where the firm defaults if the asset value $A_t$ dips below a critical threshold $J$ for any time $t\in (0,\infty]$, or if the terminal asset value, $A_T$ is less than $K$. We assume that $J<K$. Define the following stopping times
	\begin{equation}
		\tau_1=\begin{cases}
			T,\ \text{if } A_T<K\\
			\infty,\ \text{otherwise},
		\end{cases}
	\end{equation}
	and $\tau_2=\inf\{t\in(0,T]|A_t<J\}$. If $A_t$ never drops below $J$, we set $\tau_2=\infty$. Then the default time, $\tau$, is given by
	\begin{equation}
		\tau=\min(\tau_1,\tau_2).
	\end{equation}
	If the default time is infinity, the firm does not default and the bondholders receive their principal entirely. We can write the value of the defaultable bond at time $T$ as
	\begin{equation}
		D(T,A_T,X_T)=K-(K-A_T)^+ + (A_T-K)^+ \mathds{1}(\min_{t\leq T}A_t<J).\label{model 2}
	\end{equation}
	
	\noi The above payoff can at once be recognised as that of a portfolio consisting of the following three components:
	\begin{enumerate}
		\item A default-free loan of face value $ K $, with maturity $T$,
		\item A short European put option on $A_t$ with dividend rate $ \kappa(t,X_t) $, strike price $K$ and maturing at time $T$, and
		\item A long European down-and-out call option with strike price $K$, barrier $J$ and maturing at time $T$.
	\end{enumerate}
	The value of the defaultable bond under this model is at least as much as that under Merton's classical model, due to the presence of the third term in \eqref{model 2}. The bondholders are thus better protected. If the volatility does not depend explicitly on time, i.e. if $ \sigma(t,i)=\sigma(i) $ for all $ t $ and $ i $, then the pricing and hedging problems may be addressed using our analysis in \ref{barrier options}.
	
	\subsection{Model 3}
	In this model, the criteria for a default are the same as that for Model 2. The recovery rule, however, is different. In case of a premature default, the bondholders are paid a fraction of the face value of the bond at a pre-determined constant recovery rate, $\delta$, which satisfies the following inequality
	\begin{equation}
		0\leq \delta\leq \frac{J}{K} \left(\leq 1\right).
	\end{equation}
	The procedure for debt recovery is the same as that in Model 2 if the firm defaults at maturity. If the firm does not default, the debt is paid of entirely at maturity. The value of the defaultable bond at maturity can thus be written as
	\begin{equation}
		D(T,A_T,X_T)=\min(A_T,K)\mathds{1}(\tau\geq T) + \delta K B(\tau, T, X_{\tau})\mathds{1}(\tau< T),
	\end{equation}
	where $B(\tau, T, X_{\tau})$ denotes the price at time $\tau$ of a default-free couponless bond with unit face value and maturity $T$. This model is different from the two models previously discussed in that the recovery is at the time of the default, and not necessarily strictly at maturity. As in Model 2, an integral equation formalism can be used in the case where the volatility has no explicit time-dependence.\\
	
	\noi The market we are considering is incomplete (i.e not all contingent claims can be perfectly hedged by self-financing strategies). This is due to the presence of semi-Markov modulated regime switching. We can, however, minimize the residual risk arising from the incompleteness of the market. We look for the price of derivative securities that minimizes the residual risk. This can be done by considering the F\"{o}llmer-Schweizer decomposition of the relevant contingent claim.

\end{document}